%% file: main.tex
\RequirePackage[l2tabu,orthodox]{nag}

\documentclass[headsepline,footsepline,footinclude=false,oneside,fontsize=11pt,paper=a4,listof=totoc,bibliography=totoc]{scrbook} 

\PassOptionsToPackage{table,svgnames,dvipsnames}{xcolor}

\usepackage[utf8]{inputenc}
\usepackage[T1]{fontenc}
\usepackage[sc]{mathpazo}
\usepackage[ngerman,american]{babel}
\usepackage[autostyle]{csquotes}
\usepackage[%
  backend=biber,
  url=false,
  style=alphabetic,
  maxnames=4,
  minnames=3,
  maxbibnames=99,
  giveninits,
  uniquename=init]{biblatex}
\usepackage{graphicx}
\usepackage{scrhack} 
\usepackage{listings}
\usepackage{lstautogobble}
\usepackage{tikz}
\usepackage{pgfplots}
\usepackage{pgfplotstable}
\usepackage{booktabs}
\usepackage{caption}
\usepackage[printonlyused]{acronym}
\usepackage{longtable}
\usepackage{amsmath}
\usepackage{amssymb}
\usepackage{amsthm}
\usepackage{geometry}
\usepackage[hidelinks]{hyperref} 
\AtBeginDocument{%
	\hypersetup{
		pdftitle=\getTitle,
		pdfauthor=\getAuthor,
	}
}
\usepackage{ifthen}

\addto\extrasamerican{

}

\addto\extrasngerman{
	
}

\ifthenelse{\equal{\detokenize{dark}}{\jobname}}{%
  \newcommand{\bg}{black} 
  \newcommand{\fg}{white} 
  \usepackage[pagecolor=\bg]{pagecolor}
  \color{\fg}
}{%
  \newcommand{\bg}{white} 
  \newcommand{\fg}{black} 
}

\bibliography{bibliography}

\setkomafont{disposition}{\normalfont\bfseries} 
\BeforeTOCHead[toc]{{\cleardoublepage\pdfbookmark[0]{\contentsname}{toc}}}

\definecolor{TUMBlue}{HTML}{0065BD}
\definecolor{TUMSecondaryBlue}{HTML}{005293}
\definecolor{TUMSecondaryBlue2}{HTML}{003359}
\definecolor{TUMBlack}{HTML}{000000}
\definecolor{TUMWhite}{HTML}{FFFFFF}
\definecolor{TUMDarkGray}{HTML}{333333}
\definecolor{TUMGray}{HTML}{808080}
\definecolor{TUMLightGray}{HTML}{CCCCC6}
\definecolor{TUMAccentGray}{HTML}{DAD7CB}
\definecolor{TUMAccentOrange}{HTML}{E37222}
\definecolor{TUMAccentGreen}{HTML}{A2AD00}
\definecolor{TUMAccentLightBlue}{HTML}{98C6EA}
\definecolor{TUMAccentBlue}{HTML}{64A0C8}

\pgfplotsset{compat=newest}
\pgfplotsset{
  cycle list={TUMBlue\\TUMAccentOrange\\TUMAccentGreen\\TUMSecondaryBlue2\\TUMDarkGray\\},
}

\newtheorem{theorem}{Theorem}[section]
\newtheorem{corollary}[theorem]{Corollary}
\newtheorem{lemma}[theorem]{Lemma}
\theoremstyle{definition}
\newtheorem{definition}[theorem]{Definition}
\theoremstyle{remark}
\newtheorem*{remark}{Remark}
\newtheorem*{example}{Example}

\definecolor{isarblue}{HTML}{006699}
\definecolor{isargreen}{HTML}{009966}
\definecolor{isarlightblue}{HTML}{0099ff}
\lstdefinelanguage{isabelle}{%
    keywords=[1]{locale,sublocale,global_interpretation,interpretation,context,definition,proof,lemma,proposition,theorem,corollary,next,qed,using,unfolding,by,have,moreover,ultimately,hence},
    keywordstyle=[1]\bfseries\color{isarblue},
    keywords=[2]{where,fixes,obtains,assumes,shows,and,begin,end,for},
    keywordstyle=[2]\bfseries\color{isargreen},
    keywords=[3]{if,then,else,SOME,let,in,AE,LINT},
    keywordstyle=[3]\color{isarblue},
	keywords=[4]{define,assume,fix,show,thus, obtain},
    keywordstyle=[4]\bfseries\color{isarlightblue},
}
\lstdefinestyle{isabelle}
{%
  language=isabelle,
  escapeinside={&}{&},
  columns=fixed,
  extendedchars,
  basewidth={0.5em,0.45em},
  basicstyle=\ttfamily,
  mathescape,
  tabsize=4
}

\newtheoremstyle{break}
  {\topsep}{\topsep}%
  {}{}%
  {\bfseries}{}%
  {\newline}{}%
\theoremstyle{break}

\newtheorem{isacorollary}[theorem]{Corollary}
\newtheorem{isalemma}[theorem]{Lemma}
\newtheorem{isadefinition}[theorem]{Definition}

\newcommand*{\getUniversity}{Technische Universität München}
\newcommand*{\getFaculty}{Informatics \& Mathematics}

\newcommand*{\getSchool}{Computation, Information and Technology}
\newcommand*{\getTitle}{On the Formalization of Martingales}
\newcommand*{\getTitleGer}{Eine Formalisierung von Martingalen}
\newcommand*{\getAuthor}{Ata Keskin}
\newcommand*{\getDoctype}{Bachelor's Thesis}
\newcommand*{\getSupervisor}{Prof. Dr. Tobias Nipkow}
\newcommand*{\getAdvisor}{M. Sc. Katharina Kreuzer}
\newcommand*{\getSubmissionDate}{15 September 2023}

\begin{document}

\pagenumbering{alph}

\frontmatter{}

\input{title}
\input{abstract}
\tableofcontents{}

\mainmatter{}

\input{01_introduction}
\input{02_background}
\input{03_conditional_expectation}
\input{04_stochastic_processes}
\input{05_martingales}
\input{06_discussion}

\input{07_conclusion}

\printbibliography{}

\appendix
\input{appendix}

\end{document}

%% file: title.tex
\begin{titlepage}
  \centering

  \includegraphics[height=20mm]{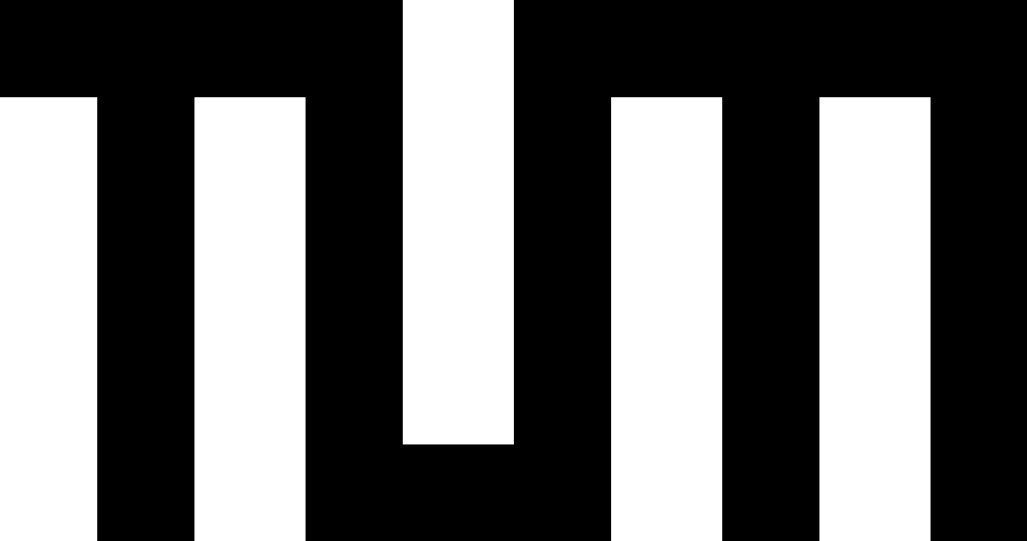}

  \vspace{5mm}
  {\huge\MakeUppercase{School of \getSchool{} --- \getFaculty{}} \par}

  \vspace{5mm}
  {\large\MakeUppercase{\getUniversity{}} \par}

  \vspace{20mm}
  {\Large \getDoctype{} in Informatics \par}

  {\Large \getDoctype{} in Mathematics \par}
  \vspace{8mm}
  
  \includegraphics[height=30mm]{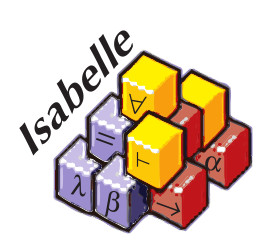}

  \vspace{8mm}
  {\huge\bfseries \getTitle{} \par}

  \vspace{10mm}
  {\huge\bfseries \foreignlanguage{ngerman}{\getTitleGer{}} \par}

  \vspace{10mm}
  \begin{tabular}{l l}
    Author:          & \getAuthor{}         \\
    Supervisor:      & \getSupervisor{}     \\
    Advisor:         & \getAdvisor{}        \\
    Submission Date: & \getSubmissionDate{} \\
  \end{tabular}

\end{titlepage}

%% file: abstract.tex
\chapter{\abstractname}

This thesis presents a formalization of martingales in arbitrary Banach spaces using Isabelle/HOL. We begin by examining formalizations in prominent proof repositories and extend the definition of the conditional expectation operator from the real numbers to general Banach spaces, drawing inspiration from prior work. We define filtered measure spaces, adapted, progressively measurable and predictable processes and rigorously formalize martingales, submartingales, and supermartingales. Additionally, our contributions expand the scope of Bochner integration techniques to general Banach spaces, and introduces additional lemmas and induction schemes for integrable functions. Our formalization provides a robust framework for future formalizations within the theory of stochastic processes.

%% file: 01_introduction.tex

\chapter{Introduction}\label{chapter:introduction}

Martingales hold a central position in the theory of stochastic processes, making them a fundamental concept for the working mathematician. They provide a powerful way to study and analyze random phenomena, offering a mathematical framework for understanding their behavior. 

In various real-world scenarios, we encounter systems that evolve randomly over time which can be modeled using martingales. In finance and economics, martingales are an invaluable tool for modeling asset prices \cite{fama1965} and option pricing \cite{Musiela_Rutkowski_2005}. They provide insight into risk assessment, portfolio management, and the efficient market hypothesis, which postulates that asset prices fully reflect all available information \cite{yaes1989}. 

Martingales are also closely related to several important limit theorems in probability theory. These theorems, such as the strong law of large numbers and the central limit theorem, formalize the asymptotic behavior of sample means and sums of random variables. They have profound implications in statistics, allowing us to draw conclusions about large datasets and make predictions based on limited information. Martingale theory allows us to investigate whether these systems remain bounded or converge to certain values in the long run.

In addition to their relevance in mathematics, martingales find applications in various interdisciplinary fields. Their ability to model randomness and analyze dynamic systems makes them useful in physics \cite{roldan2023}, biology, and computer science \cite{mitzenmacher_upfal_2005}, among others.

In the scope of this thesis, we present a formalization of martingales in arbitrary Banach spaces using Isabelle/HOL \cite{Keskin_A_Formalization_of_2023}. We start our discourse by examining existing formalizations in two prominent formal proof repositories, the Lean Mathematical Library (\textsf{mathlib}) and the Archive of Formal Proofs (\textsf{AFP}). Afterwards, we go over some of the basic concepts concerning the theory of integration in Banach spaces, laying a solid foundation for our research.

The current formalization of conditional expectation in the Isabelle library is limited to real-valued functions. To overcome this limitation, we extend the construction of conditional expectation to general Banach spaces, employing an approach similar to the one described in \cite{Hytoenen_2016}. We use measure theoretic arguments to construct the conditional expectation using simple functions and limiting arguments. We compare our construction\footnote{\texttt{Martingale.Conditional\_Expectation\_Banach} in \cite{Keskin_A_Formalization_of_2023}} with the approach in \cite{Hytoenen_2016} in Chap. \ref{chapter:discussion}.

Subsequently, we define stochastic processes and introduce the concepts of adapted, progressively measurable and predictable processes using suitable locale definitions\footnote{\texttt{Martingale.Stochastic\_Process} in \cite{Keskin_A_Formalization_of_2023}}. We pay special attention to predictable processes in discrete-time, providing a characterization using adapted processes. Moving forward, we rigorously define martingales, submartingales, and supermartingales, presenting their first consequences and corollaries\footnote{\texttt{Martingale.Martingale} in \cite{Keskin_A_Formalization_of_2023}}. Discrete-time martingales are given special attention in the formalization. In every step of our formalization, we make extensive use of the powerful locale system of Isabelle.
Our formalization fully encompasses the introductory \textsf{mathlib} theory \texttt{probability.mar\-tingale.basic} on martingales \cite{Degenne_Ying_2022}, even offering more generalization at certain stages.

Our thesis further contributes by generalizing concepts in Bochner integration by extending their application from the real numbers to arbitrary Banach spaces equipped with a second-countable topology. Induction schemes for integrable simple functions on Banach spaces are introduced, accommodating various scenarios with or without a real vector ordering\footnote{\texttt{Martingale.Bochner\_Integration\_Addendum} in \cite{Keskin_A_Formalization_of_2023}}. These amendments expand the applicability of Bochner integration techniques.
We conclude our thesis with reflections on the formalization approach and suggestions for future research directions.

%% file: 02_background.tex

\chapter{Background and Related Work}\label{chapter:background}

In the following section, we explore existing formalizations of martingales within the mathematical proof repositories \textsf{mathlib} and \textsf{AFP}. Afterwards, we will provide a concise introduction to the theory of integration in Banach spaces, establishing the mathematical foundation that underpins our formalization efforts.

\section{Existing Formalizations}

We start by looking at the existing developments in the proof repositories \textsf{mathlib} and \textsf{AFP}.

\subsection{Lean Mathematical Library}

Our main motivation for formalizing a theory of martingales in Isabelle/HOL comes from the existing in-depth formalization of the same subject in \textsf{mathlib}. As stated on their online platform, ``The Lean mathematical library, \textsf{mathlib}, is a community-driven effort to build a unified library of mathematics formalized in the Lean proof assistant'' \cite{Lean}. The Lean formalization of martingales consists of six documents. In the introductory Lean document \texttt{probability.martingale.basic}, fundamentals of the theory of martingales are formalized \cite{Degenne_Ying_2022}. The aim of this bachelor's thesis is to reproduce the results contained in this document using Isabelle/HOL. As will become clear in a moment, this is not a straightforward task, since there are a lot of dependencies missing in the Isabelle/HOL libraries.

The document \texttt{probability.martingale.basic} contains definitions for martingales, submartingales and supermartingales. The main results of this document are
\begin{itemize}
\item[$\rightarrow$]\lstinline[mathescape]{measure_theory.martingale $f$ $\mathcal{F}$ $\mu$:}
\item[] \quad$f$ is a martingale with respect to filtration $\mathcal{F}$ and measure $\mu$.
\item[$\rightarrow$]\lstinline[mathescape]{measure_theory.supermartingale $f$ $\mathcal{F}$ $\mu$:}
\item[] \quad$f$ is a supermartingale with respect to filtration $\mathcal{F}$ and measure $\mu$.
\item[$\rightarrow$]\lstinline[mathescape]{measure_theory.submartingale $f$ $\mathcal{F}$ $\mu$:}
\item[] \quad$f$ is a submartingale with respect to filtration $\mathcal{F}$ and measure $\mu$.
\item[$\rightarrow$]\lstinline[mathescape]{measure_theory.martingale_condexp $f$ $\mathcal{F}$ $\mu$:}
\item[] \quad the sequence $(\mu[f \vert \mathcal{F}_i])_{i\in\mathcal{T}}$ is a martingale with respect to $\mathcal{F}$ and $\mu$, where $\mu[f \vert \mathcal{F}_i]$ denotes the conditional expectation of $f$ with respect to the subalgebra $\mathcal{F}_i$ and $\mu$.
\end{itemize}

On a first note, we see that this theory relies heavily on the conditional expectation operator in Banach spaces. Prior to our development, the only formalization of conditional expectation in Isabelle/HOL was done in the real setting and resides in the theory document \texttt{HOL-Probability.Conditional\_Expectation}. This formalization was accomplished by S\`ebastien Gou\"ezel, presumably in anticipation of his later entries \cite{Ergodic_Theory-AFP} and \cite{Lp-AFP}. We will delve further into the existing formalization and how our contribution improves upon it in the upcoming section.

Within the \textsf{mathlib} formalization, the majority of lemmata on martingales require the measures in question to be finite. In our formalization of martingales, we will demonstrate that $\sigma$-finiteness suffices alone. This approach is also consistent with our generalized formalization of conditional expectation, as it inherits the $\sigma$-finiteness requirement from the pre-existing formalization in the real setting.

Another short-coming of the \textsf{mathlib} formalization is its treatment of predictable processes. The proof library \textsf{mathlib} contains the definition of adapted processes and progressively measurable processes. However, no explicit definition of a predictable process is given. Instead, predictability is defined only in the discrete-time case, using an equivalent characterization via adapted processes. In contrast, our formalization defines predictable processes more generally using the concept of a predictable $\sigma$-algebra. One of the major advantages of our formalization is the use of locales and sublocale relations. Concretely, we will show the following relationship between various types of stochastic processes.
\[
	\texttt{stochastic} \supseteq \texttt{adapted} \supseteq \texttt{progressive} \supseteq \texttt{predictable}
\]

Another important point to consider is the restrictions placed on the types in question. In the \textsf{mathlib} formalization, martingales are defined as a family of integrable functions $f : \iota \rightarrow \Omega \rightarrow E$, indexed by the set $\iota$. The \textsf{mathlib} formalization further requires that
\begin{itemize}
\item $\iota$ is a pre-ordered set,
\item $\Omega$ is a measurable space (i.e. a set together with a $\sigma$-algebra $\Sigma$),
\item $E$ is a normed and complete $\mathbb{R}$-vector space, i.e. a Banach space.
\end{itemize}

These restrictions are easily replicated in our formalization using type classes and the ``\texttt{measure}'' type. We simply restrict ourselves to functions $f : \; 't \;\rightarrow \; 'a \;\texttt{measure} \;\rightarrow\; 'b$, where the type $'t$ is an instance of the type class \texttt{order\_topology} and the type $'b$ is an instance of the type class \texttt{banach}. With this specification, our approach mirrors the \textsf{mathlib} formalization, since measure spaces, measurable spaces and $\sigma$-algebras are all represented using the type ``\texttt{measure}'' in Isabelle/HOL.

The main purpose of the \textsf{mathlib} formalization on martingales is to prove Doob's martingale convergence theorems, which concern discrete-time and continuous-time martingales (i.e. the naturals or the reals as indices). This justifies their focus on discrete-time processes and the formulation of predictability only in the discrete-time case. More information on the specifics and the development of Doob's martingale convergence theorems is available in \cite{ying2022formalization}.

This concludes our review of the \textsf{mathlib} formalization on martingales.

\subsection{Archive of Formal Proofs}

The Archive of Formal Proofs (\textsf{AFP}) is a digital repository of formalized proofs and theories developed using the Isabelle theorem prover and proof assistant. The \textsf{AFP}  hosts a variety of formalizations and proofs, primarily in the fields of logic, mathematics, and computer science. The repository allows researchers to share their formal proofs, theories, and related materials with the broader community.

The repository offers a search function, which allows us to find if any formalization on martingales have been done previously. A quick search yields the theory file \texttt{DiscretePricing.Martingale} \cite{DiscretePricing-AFP}. The entry \texttt{DiscretePricing}, by Mnacho Echenim, focuses on the formalization of the Binomial Options Pricing Model in finance \cite{EchenimPeltier}. A development of discrete-time real-valued martingales is given in order to introduce the concept of risk-neutral measures. Similar to the development on \textsf{mathlib}, the goal of this entry is not to formalize martingales. An incomplete formalization of martingales and filtered measure spaces is only given as a byproduct.
Apart from this entry, no other development on the theory of martingales is present on the \textsf{AFP}.

\section{Mathematical Foundations and Reference Material}

The main focus of our project is to formalize martingales in as general of a setting as possible. In this vein, we will study martingales defined on arbitrary Banach spaces, as opposed to the reals only. The main obstacle we will face is the development of the conditional expectation operator in arbitrary Banach spaces. As a primer for the upcoming chapters, we will quickly cover the basics of integration on Banach spaces. More information covering the prerequisites of our work can be found in the initial chapters of the book \textit{Analysis in Banach Spaces} \cite{Hytoenen_2016}.

\begin{remark}
	For the remainder of this document, unless stated otherwise explicitly, we fix a measure space $M = (\Omega, \Sigma, \mu)$ and a Banach space $(E, \lVert \cdot \rVert)$. Here $\Omega$ denotes an arbitrary set, $\Sigma$ a $\sigma$-algebra defined on this set, and $\mu : \Sigma \rightarrow \mathbb{R}_{\ge 0}$ a measure. Similarly, $E$ is a vector space which is complete with respect to the metric topology generated by the norm $\lVert \cdot \rVert$.
\end{remark}

Integration on Banach spaces is usually done using the Bochner-integral, which is defined similarly to the Lebesgue-integral. For $M$ a measure space and $E$ a Banach space, we introduce the Bochner-integral as follows:

\vspace{0.3cm}

We consider simple functions $s : \Omega \rightarrow E$. These are functions which can be expressed $\mu$-almost everywhere ($\mu$-a.e.) as finite sums of the form
\[
	s = \sum_{i=1}^n \mathbf{1}_{A_i} \cdot_\mathbb{R} c_i
\]
where $\mathbf{1}_A$ is the indicator function of a set $A \in \Sigma$ and $c_i \in E$. Here $\cdot_\mathbb{R}$ denotes the scalar multiplication in $E$. We call such a function $s$ Bochner-integrable if $\mu(A_i) < \infty$ for all $A_i \in \Sigma$. In this case, we define the Bochner-integral simply as the sum
\[
	\int s \; \textrm{d}\mu = \sum_{i=1}^n \mu(A_i) \cdot_\mathbb{R} c_i
\]
If we replace $E$ with $\mathbb{R}$, we can easily see that Bochner-integrable simple functions are exactly those functions, which are Lebesgue-integrable and simple.
\vspace{0.2cm}

We call a function $f : \Omega \rightarrow E$ strongly measurable, if there exists a sequence $(f_n)_{n\in\mathbb{N}}$ of simple functions converging to $f$ $\mu$-almost everywhere.
A strongly measurable function $f$ is called Bochner-integrable with respect to $\mu$, if there exists a sequence of Bochner-integrable simple functions $f_n : \Omega \rightarrow E$ such that
\[
	\lim_{n \to \infty} \int_\Omega \lVert f - f_n \rVert \; \textrm{d}\mu = 0
\]
The integral used in this definition is the ordinary Lebesgue-integral. This definition makes sense, since $w \mapsto \lVert f(w) - f_n(w) \rVert$ is $\mu$-measurable and non-negative.

It can be shown via the triangle inequality that the integrals $\int f_n\; \textrm{d}\mu$ form a Cauchy sequence. By completeness, this sequence converges to some element $\lim_{n \to \infty} \int f_n \; \textrm{d}\mu \in E$. This limit is called the Bochner-integral of $f$ with respect to the measure $\mu$
\[
	\int f \; \textrm{d}\mu = \lim_{n \to \infty} \int f_n \; \textrm{d}\mu
\]

Furthermore, a function $f$ in this setting is Bochner-integrable, if and only if the function $x \mapsto \lVert f(x) \rVert$ is integrable.

\vspace{0.3cm}

A formalization of the Bochner-integral is available in Isabelle/HOL in the theory file \texttt{HOL-Analysis.Bochner\_Integration} \cite{hoelzl2011measuretheory}. This formalization, by Johannes Hölzl, has the additional assumption that the space $E$ be second-countable. In the context of a metric space, this is the same as requiring separability.
\vspace{0.3cm}
\begin{remark}

One can show that a function $f$ is strongly measurable if and only if it is essentially separably valued and for all $A \in \mathcal{B}(E)$ we have $f^{-1}(A) \in \Sigma$. Here $\mathcal{B}(E)$ denotes the Borel $\sigma$-algebra on $E$. A function is called essentially separably valued if there exists a $\mu$-null set $N \subseteq \Omega$, such that $f (\Omega \setminus N)$ is separable as a subspace of $E$. Therefore, if $E$ is already a separable Banach space, a function $f : \Omega \rightarrow E$ is strongly measurable if and only if it is $\Sigma$-measurable.

Consequently, we do not need to concern ourselves with definining strong measurability when working within separable (or equivalently second-countable) Banach spaces.
\end{remark}
\vspace{0.3cm}

The book \textit{Analysis in Banach Spaces} also contains an in depth section on the construction of the conditional expectation operator on Banach spaces. For our purposes, we only need to focus on the case where $f : \Omega \rightarrow E$ is a Bochner-integrable function, i.e. an element of $L^1(E)$. Here, $L^p(E)$ denotes the set of functions $f: \Omega \rightarrow E$, for which $x \mapsto \lVert f(x) \rVert^p$ is integrable which are unique upto a $\mu$-null set. In this case, the conditional expectation is constructed as a linear operator $L^1(E) \rightarrow L^1(E)$. The book contains theorems for the existence and uniqueness of conditional expectations (up to $\mu$-null sets) for functions not only in $L^1(E)$, but also for those in $L^2(E)$ and $L^0(E)$. The latter is the space of strongly measurable functions with codomain $E$. Unsuprisingly, the definition of conditional expectation in this case is a bit more complicated, since it has to take into account the case where $f$ is not integrable.

To demonstrate commonly used properties of conditional expectation, we have drawn upon the ideas presented in the lecture notes of Gordan Zitkovic during his lecture on conditional expectation at the University of Texas at Austin \cite{Zitkovic_2015}.

For defining stochastic processes in a general setting, we have used the definitions presented in the books \textit{PDE and Martingale Methods in Option Pricing} by Andrea Pascucci \cite{Pascucci_2011} and \textit{Stochastic Calculus and Applications} by Samuel N. Cohen and Robert J. Elliott \cite{Elliott_Cohen_1982}. Apart from these resources, we have made heavy use of the blog \textit{Almost Sure} by George Thowler from the University of Cambridge \cite{Thowler}.

Another extensive reference regarding martingales in Banach spaces is the book \textit{Martingales in Banach Spaces} by Gilles Pisier \cite{pisier_2016}. This resource provides an in-depth exploration of the theory of martingales in Banach spaces at a graduate level. Given the limited scope of this thesis, the book serves only as a supplementary resource.

%% file: 03_conditional_expectation.tex

\chapter{Conditional Expectation in Banach Spaces}\label{chapter:conditional_expectation}

Conditional expectation extends the concept of expected value to situations where we have additional information about the outcomes. In a discrete setting, i.e. when the range of the random variables in question is countable, the setup is quite simple. 
\begin{example}
\par Let $(\Omega, \mathcal{F}, \mu)$ be a measure space. Let $E$ be a complete normed vector space, i.e. a Banach space, and $S \subseteq E$ be some countable subset. Let $X : \Omega \rightarrow S$ be an integrable random variable and an event $A \in \mathcal{F}$ with $\mu(A) < \infty$. The conditional expectation of $X$ given $A$, denoted as $\mathbb{E}(X \vert A)$, represents the expected value of $X$ given that $A$ occurs. In this simple (and countable) case, we can directly define the conditional expectation as:
\[
	\mathbb{E}(X \vert A) = \sum_{w \in S} \frac{\mu(\{X = w\} \cap A)}{\mu(A)} \cdot w
\]
\end{example}
Of course, this definition only makes sense if the value on the right hand side is finite and $\mu(A) \neq 0$. Defined this way, the conditional expectation satisfies the following equality
\begin{align*}
	\int_A X \; \textrm{d} \mu &= \sum_{w \in S} \mu(\{\mathbf{1}_A \cdot X = w\}) \cdot w\\
	&= \mu(A) \cdot \mathbb{E}(X \vert A) \\
	&= \int_A \mathbb{E}(X \vert A) \; \textrm{d} \mu
\end{align*}

\begin{remark}
	We use the notation ``$c \; \cdot \; w$'' to denote the scalar multiplication of $c \in \mathbb{R}$ and $w \in E$. When $E = \mathbb{R}$, it is just the standard multiplication on $\mathbb{R}$. 
\end{remark}

This observation motivates us to generalize the definition of conditional expectation to take into account not just a single event, but a collection of events. 

Fix $X : \Omega \rightarrow E$. Given a sub-$\sigma$-algebra $\mathcal{H} \subseteq \Sigma$, we call an $\mathcal{H}$-measurable function $g : \Omega \rightarrow E$ a conditional expectation of $X$ with respect to the sub-$\sigma$-algebra $\mathcal{H}$ , denoted as $\mathbb{E}(X \vert \mathcal{H})$, if the following equality holds for all $A \in \mathcal{H}$

\[
	\int_A X \; \textrm{d} \mu = \int_A g \; \textrm{d} \mu
\]

In the case that $E = \mathbb{R}$, it is straightforward to show that such a function $g$ always exists (via Radon-Nikodym), and is unique up to a $\mu$-null set. Notice that $\mathbb{E}(X \vert \mathcal{H})$ is a function $\Omega \rightarrow E$, as opposed to some value in $E$.

The suitable setting for defining the conditional expectation is when the sub-$\sigma$-algebra $\mathcal{H}$ gives rise to a $\sigma$-finite measure space. This is the case when $\mu\vert_\mathcal{H}$, the restriction of $\mu$ to $\mathcal{H}$, is a $\sigma$-finite measure. To see what goes wrong otherwise, consider the trivial sub-$\sigma$-algebra $\{\varnothing, \Omega\}$. A function which is measurable with respect to this $\sigma$-algebra is necessarily constant. Therefore, if $\mu(\Omega) = \infty$, no conditional expectation can exist, since it would have to be equal to $0$ $\mu$-almost everywhere in order to be integrable.

\vspace{0.3cm}
\begin{example}
Let $\mathcal{H} \subseteq \mathcal{F}$ be a sub-$\sigma$-algebra such that $\mu\vert_\mathcal{H}$ is a $\sigma$-finite measure. Given an integrable function $X : \Omega \rightarrow \mathbb{R}$, we can define a measure $\nu$ on $(\Omega, \mathcal{F})$ via
\[
	\nu(A) := \int_A X \; \textrm{d}\mu
\]
It is easy to verify that $\mu\vert_\mathcal{H}(A) = 0$ implies $\nu\vert_\mathcal{H}(A) = 0$, i.e. $\nu\vert_\mathcal{H}$ is absolutely continuous with respect to $\mu\vert_\mathcal{H}$. Using the Radon-Nikodym Theorem, we obtain an $\mathcal{H}$-measurable function $g : \Omega \rightarrow \mathbb{R}$ such that
\[
	\nu\vert_\mathcal{H}(A) = \int_A g \;\textrm{d}\mu\vert_\mathcal{H}
\]
Thus for any $A \in \mathcal{H}$, we have
\[
	\int_A X \; \textrm{d}\mu = \int_A g \;\textrm{d}\mu\vert_\mathcal{H} = \int_A g \;\textrm{d}\mu
\]
In the last equality, we use the fact that $g$ is $\mathcal{H}$-measurable. Radon-Nikodym also guarentees that this function $g$ is unique up to a $\mu\vert_\mathcal{H}$-null set. Since all $\mu\vert_\mathcal{H}$-null sets are also $\mu$-null sets, the function $g$ satisfies the requirements of a conditional expectation.
\end{example}
\vspace{0.3cm}

Details aside, this shows that the conditional expectation always exists and is unique up to $\mu$-null set for all $X \in L^1(\mathbb{R})$. Our job now will be to construct a similar operator on arbitrary Banach spaces using methods from functional analysis and measure theory. The results formalized in this chapter can be found in the theory files \texttt{Martingale.Elementary\_Metric\_Spaces\_Supplement}, \;\texttt{Martingale.Bochner\_Integra\-tion\_Supplement}, and \texttt{Martingale.Conditional\_Expectation\_Banach} in \cite{Keskin_A_Formalization_of_2023}.

\section{Preliminaries}

In anticipation of our construction, we need to lift some results from the real setting to our more general setting. Our fundamental tool in this regard will be the \textbf{averaging theorem}. The proof of this theorem is due to Serge Lang \cite{Lang_1993}. The theorem allows us to make statements about a function's value almost everywhere, depending on the value its integral takes on various sets of the measure space.

\subsection{Averaging Theorem}

Before we introduce and prove the averaging theorem, we will first show the following lemma which is crucial for our proof. While not stated exactly in this manner, our proof makes use of the characterization of second-countable topological spaces given in the book General Topology by Ryszard Engelking (Theorem 4.1.15) \cite{engelking_1989}.

\begin{lemma}
Let $E$ be a separable metric space. Then there exists a countable set $D \subseteq E$, such that the set of open balls
\[
	\mathcal{B} = \{ B_\varepsilon(x) \; \vert \; x \in D, \; \varepsilon \in \mathbb{Q} \cap (0, \infty) \}
\]
generates the topology on $E$. Here $B_\varepsilon(x)$ is the open ball of radius $\varepsilon$ with centre $x$.
\end{lemma}

\begin{proof}
In the context of metric spaces, second-countability is equivalent to separability. Consequently, there exists some non-empty countable subset $D \subseteq E$, which is dense in $E$. We want to show that this $D$ fulfills the statement above. To this end we will use the following equivalence which is valid for any $\mathcal{A} \subseteq \mathcal{P}(E)$

\[
	\mathcal{A} \textrm{ is topological basis} \Longleftrightarrow \forall \textrm{open } U.\; \forall x \in U.\; \exists A \in \mathcal{A}.\; x \in A \wedge A \subseteq U
\]

Let $U \subseteq E$ be open. Fix $x \in U$. Since $U$ is open and we are working with the metric topology, there is some $\varepsilon > 0$, such that $B_\varepsilon(x) \subseteq U$. Furthermore, we know that a set $D$ is dense if and only if for any non-empty open subset $O \subseteq E$, $D \cap O$ is also non-empty. Therefore, there exists some $y \in D \cap B_{\varepsilon/3}(x)$. Since $\mathbb{Q}$ is dense in $\mathbb{R}$, there exists some $r \in \mathbb{Q}$ with $e/3 < r < e/2$. It is easy to check that $x \in B_r(y)$ and $B_r(y) \subseteq U$ with $y \in D$ and $r \in \mathbb{Q} \cap (0, \infty)$. This concludes the proof.
\end{proof}
Now we are ready to state and subsequently prove the averaging theorem.

\begin{theorem}\label{thm:averaging_theorem} (Averaging Theorem) \par
Let $(\Omega, \mathcal{F}, \mu)$ be some $\sigma$-finite measure space. Let $f \in L^1(E)$. Let $S$ be a closed subset of $E$ and assume that for all measurable sets $A \in \mathcal{F}$ with finite and non-zero measure the following holds
\[
	\frac{1}{\mu(A)}\int_A f \;\textrm{d}\mu \in S
\]
Then $f(x) \in S$ for $\mu$-almost all $x$.
\end{theorem}
\begin{proof}
Without loss of generality we will show the statement assuming $\mu(\Omega) < \infty$. Let $v \in E$ and $v \notin S$. 

We show by contradiction that if $B_r(v) \cap S = \varnothing$,  then $A := f^{-1}(B_r(v))$, the set of all $x \in \Omega$ such that $f(x) \in B_r(v)$, is a $\mu$-null set. Therefore, let's assume $\mu(A) > 0$. We have

\begin{align*}
	\left\lVert \frac{1}{\mu(A)}\int_A f \;\textrm{d}\mu  - v \right\rVert &= \left\lVert \frac{1}{\mu(A)}\int_A f - v \;\textrm{d}\mu \right\rVert \\
	&\le \frac{1}{\mu(A)}\int_A \lVert f - v \rVert \;\textrm{d}\mu \\
	&< r
\end{align*}

The last inequality follows from the fact that $f(x) \in B_r(v)$ for $x \in A$. This contradicts our first assumption. Therefore $\mu(A) = 0$.

Notice that $E \setminus S$ is an open subset of $E$. By the previous lemma, there exist open balls $B_{r_i}(w_i)$ with $r_i \in \mathbb{Q}_{\ge 0}$, $w_i \in D$ for $i \in \mathbb{N}$ such that $\bigcup_i B_{r_i}(w_i) = - S$. Obviously, $w_i \in E \setminus S$ and $B_{r_i}(w_i) \cap S = \varnothing$ for $i \in \mathbb{N}$. It follows

\begin{align*}
	\mu(f^{-1}(E \setminus S)) &= \mu\left(\bigcup_i f^{-1}(B_{r_i}(w_i))\right) \\
	&\le \sum_i \mu(f^{-1}(B_{r_i}(w_i))) \\
	&= 0
\end{align*}

Thus $\{f \notin S \}$ is a $\mu$-null set, which completes the proof.

\end{proof}

The formalization of the averaging theorem was part of our development\footnote{\texttt{Bochner\_Integration\_Supplement.averaging\_theorem} in \cite{Keskin_A_Formalization_of_2023}}. At the beginning of our proof, we assumed $\mu(\Omega) < \infty$ without loss of generality. This is only possible since we assumed the measure space in question to be $\sigma$-finite. To simplify the formalization of proofs employing this argument, we have introduced the following induction scheme

\begin{isalemma}
{\small
	\begin{lstlisting}[style=isabelle]
	lemma sigma_finite_measure_induct:
	  assumes "$\bigwedge N \; \Omega. \;\; \texttt{finite\_measure} \; N$
			$\Longrightarrow N = \texttt{restrict\_space} \; M \; \Omega$
			$\Longrightarrow \Omega \in \texttt{sets} \; M$
			$\Longrightarrow \texttt{emeasure}\; N \; \Omega \neq \infty $
			$\Longrightarrow \texttt{emeasure}\; N \; \Omega \neq 0$
			$\Longrightarrow \texttt{almost\_everywhere} \; N \; Q$"
	  and "$\texttt{Measurable.pred} \; M \; Q$"
	  shows "$\texttt{almost\_everywhere} \; M \; Q$"
	\end{lstlisting}
}
\end{isalemma}

This induction scheme allows us to prove results about a $\sigma$-finite measure space $M$, assuming that we can show the property on arbitrary subspaces of $M$ with finite measure. For increased usability, we include additional assumptions such as $\texttt{emeasure}\; N \; \Omega \neq 0$ which let us disregard trivial measure spaces. The proof of this induction scheme is straightforward.
\begin{proof}
Let $M = (\Omega, \Sigma, \mu)$ be a $\sigma$-finite measure space. There exists a family of sets with finite measure $(\Omega_i)_{i \in \mathbb{N}}$ such that $\bigcup_{i \in \mathbb{N}} \Omega_i = \Omega$. By assumption, the property $Q$ holds $\mu$-almost everywhere on all $\Omega_i$. Therefore the sets $\Omega_i \cap \{x \in \Omega \;\vert\; \neg Q(x)\} \in \Sigma\vert_{\Omega_i} \subseteq \Sigma$ are all $\mu$-null sets. This means that $\bigcup_{i \in \mathbb{N}} (\Omega_i \cap \{x \in \Omega \;\vert\; \neg Q(x)\}) = \{x \in \Omega \;\vert\; \neg Q(x)\}$ is also a $\mu$-null set, which completes the proof.
\end{proof}

Now that we have the averaging theorem at our disposal, we can lift the following results from the real case, to our more general setting.

\begin{corollary}
	Let $f \in L^1(E)$ and $\int_A f \;\textrm{d}\mu = 0$ for all measurable sets $A \subseteq \Omega$. Then $f = 0$ $\mu$-almost everywhere.
\end{corollary}
\begin{proof}
	Apply the averaging theorem with $S = \{0\}$.
\end{proof}

\begin{corollary}\label{cor:density_unique} (Uniqueness of Densities) \par
	Let $f, g \in L^1(E)$ and $\int_A f \;\textrm{d}\mu = \int_A g \;\textrm{d}\mu$ for all measurable sets $A \subseteq \Omega$. Then $f = g$ $\mu$-almost everywhere.
\end{corollary}
\begin{proof}
	Follows directly from the previous corollary.
\end{proof}

\begin{corollary}
	Let $E$ be linearly orderable. Let $f \in L^1(E)$ and $\int_A f \;\textrm{d}\mu \ge 0$ for all measurable sets $A \subseteq \Omega$. Then $f$ is non-negative $\mu$-almost everywhere.
\end{corollary}
\begin{proof}
	Our first assumption guarantees that $\{ y \in E \;\vert\; y \ge 0 \}$ is a closed subset of $E$. Applying the averaging theorem on this set, yields the desired result.
\end{proof}

The corollary on the uniqueness of densities (\ref{cor:density_unique}) is crucial in showing that the conditional expectation is unique as an element of $L^1(E)$. These statements are formalized seperately as well.\footnote{\texttt{Bochner\_Integration\_Supplement.density\_zero}\quad \texttt{-.density\_unique\_banach}\quad \texttt{-.density\_nonneg} in \cite{Keskin_A_Formalization_of_2023}}

\subsection{Diameter Lemma}

The goal of this subsection is to prove the diameter lemma, which provides a characterization of Cauchy sequences in metric spaces.

\begin{definition}
Let $E$ be a metric space with metric $d : E \times E \rightarrow \mathbb{R}$. The diameter of a set $A \subseteq E$ is defined as
\[
	\textrm{diam}(A) = \sup_{x,y \in A} d(x,y)
\]
\end{definition}

Intuitively the diameter of a set $A$ measures how ``spread out'' the set $A$ is with respect to the distance defined by the metric.

\begin{lemma}\label{lem:diameter} (Diameter Lemma)\par
	Let $E$ be a metric space with metric $d : E \times E \rightarrow \mathbb{R}$ and $(s_i)_{i\in\mathbb{N}} \subseteq E$ a sequence. Define $S_n = \{s_i \; \vert \; i \ge n \}$. The sequence $(s_i)_{i\in\mathbb{N}}$ is Cauchy, if and only if $S_0$ is bounded and
	\[
		\lim_{n \to \infty} \textup{diam}(S_n) = 0
	\]
\end{lemma}
\begin{proof}
First, assume $(s_i)_{i\in\mathbb{N}}$ is Cauchy. 

Recall that a set $A$ is bounded if there exists some $x \in E$ and $\varepsilon \in \mathbb{R}$ such that $d(x,y) \le \varepsilon$ for all $y \in A$. Since $(s_i)_{i\in\mathbb{N}}$ Cauchy, there exists some $K \in \mathbb{N}$ such that $d(s_n,s_m) < 1$ for all $n, m \ge K$. The set $\{s_i \; \vert \; i \in \{0,\dots,K\}\}$ is bounded since it is finite. Thus there exists some $a \in \mathbb{R}$ such that $d(s_K, s_i) < a$ for all $i \in \{0,\dots,K\}$. Therefore $d(s_N, s_i) < \max(a, 1)$ for all $i \in \mathbb{N}$, which shows that $S_0$ is bounded. 

We know $S_n \subseteq S_m$ for $n \ge m$. Therefore $\textrm{diam}(S_n) < \infty$ for all $n \in \mathbb{N}$.

Let $\varepsilon > 0$. Then there exists some $N \in \mathbb{N}$ such that $d(s_n,s_m) < \frac{\varepsilon}{2}$ for all $n, m \ge N$. Hence
\[
	\textrm{diam}(S_N) = \sup_{x,y \in S_N} d(x,y) \le \frac{\varepsilon}{2} < \varepsilon
\]
Furthermore, we have $\textrm{diam}(S_n) \le \textrm{diam}(S_N)$ for $n \ge N$ because of the subset relation stated above. Thus $\lim_{n \to \infty} \textup{diam}(S_n) = 0$.\smallskip

For the other direction, assume $\lim_{n \to \infty} \textup{diam}(S_n) = 0$ and that $S_0$ is bounded. Hence $\textrm{diam}(S_n) < \infty$ for all $n \in \mathbb{N}$ with the same argument as above. 

Let $\varepsilon > 0$. There exists some $N \in \mathbb{N}$ such that $\sup_{x,y \in S_n} d(x,y) < \varepsilon$ for all $n \ge N$. Hence $d(x,y) < \varepsilon$ for all $x, y \in S_n$ for $n \ge N$. This implies $d(s_i,s_j) < \varepsilon$ for all $i, j \ge n \ge N$, which shows that $(s_i)_{i\in\mathbb{N}}$ is Cauchy.
\end{proof}

In our construction of the conditional expectation, we will use the diameter lemma\footnote{\texttt{Elementary\_Metric\_Spaces\_Addendum.cauchy\_iff\_diameter\_tends\_to\_zero\_and\_bounded} in \cite{Keskin_A_Formalization_of_2023}} (\ref{lem:diameter}) to show that the limit of a sequence of simple functions admits a conditional expectation. 

In anticipation of this, we formalize the following lemmas concerning measurability and integrability.

\begin{isalemma}
{\small
	\begin{lstlisting}[style=isabelle]
	lemma borel_measurable_diameter: 
	  assumes "$\bigwedge x. \;x \in \texttt{space} \; M \implies \texttt{bounded} \; (\texttt{range} \; (\lambda i. \; s \; i \; x))$"
			  "$\bigwedge i. \; (s \; i) \in \texttt{borel\_measurable} \; M$"
	  shows "$(\lambda x. \; \texttt{diameter} \; \{s \; i \; x \; \vert \; i. \; n \le i \}) \in \texttt{borel\_measurable} \; M$"
  	\end{lstlisting}
}
\end{isalemma}

\begin{isalemma}
{\small
	\begin{lstlisting}[style=isabelle]
	lemma integrable_bound_diameter: 
	  assumes "$\texttt{integrable} \; M \; f$" 
			  "$\bigwedge i. \; (s \; i) \in \texttt{borel\_measurable} \; M$"
			  "$\bigwedge x \; i. \; x \in \texttt{space} \; M \implies \texttt{norm} \; (s \; i \; x) \le f \; x$"
	  shows "$\texttt{integrable} \; M \; (\lambda x. \; \texttt{diameter} \; \{s \; i \; x \; \vert \; i. \; n \le i\})$"
  	\end{lstlisting}
}
\end{isalemma}

The proofs are straightforward and depend on the measurability of the supremum function.

\subsection{Induction Schemes for Integrable Simple Functions}

In the upcoming sections of our work, we will frequently need to prove statements about integrable simple functions. For simple functions $s : \Omega \rightarrow \mathbb{R}_{\ge 0} \cup \{\infty\}$, the Isabelle theory \texttt{HOL\_Analysis.Nonnegative\_Lebesgue\_Integration} already provides an induction scheme \texttt{simple\_function\_induct}. For our purposes we extend this scheme to cover integrable simple functions $s : \Omega \rightarrow E$. Notice that a simple function $s$ is integrable if and only if $\mu(\{s \neq 0\}) < \infty$.

The idea of the new induction scheme\footnote{\texttt{Bochner\_Integration\_Supplement.integrable\_simple\_function\_induct} in \cite{Keskin_A_Formalization_of_2023}} is simple. We know $f$ can be represented $\mu$-a.e. as a finite sum $\sum_{i=1}^n \mathbf{1}_{A_i} \cdot c_i$  for some collection of measurable sets $(A_i)_{i=1,\dots,n}$ and elements $c_i \in E$. We do induction on $n$. We first show that the statement holds for indicator functions of measurable sets with finite measure. Then, we extends this by linearity to arbitrary simple functions. Since $f$ is representable as a finite sum $\mu$-a.e. we additionally need to show that $P$ is congruent on functions which are equal $\mu$-a.e. In other words, $P(f) = P(g)$ if $f = g$ $\mu$-a.e. for any $g \in \mathcal{L}^1(E)$. This guarantees that $P$ is a well defined predicate on the space $L^1(E)$.

\begin{remark}
To make proving certain properties easier, we have the additional assumption $\lVert f(x) + g(x) \rVert = \lVert f(x) \rVert + \lVert g(x) \rVert$ in the second case of the formal statement. It is easy to see why we can assume this without loss of generality. If we have some simple function $s = \sum_{i=1}^n \mathbf{1}_{A_i} \cdot c_i$, we can assume the sets $A_i$ to be pairwise disjoint. Thus, if $x \in A_j$ for some $j \le n$ we have $\lVert s(x) \rVert = \lVert \mathbf{1}_{A_j}(x) \cdot c_j\rVert = \sum_{i=1}^n \mathbf{1}_{A_i} \cdot \lVert c_i \rVert$.
\end{remark}

When working with an ordering on $E$, we may need to concern ourselves with non-negative simple functions. For this goal, we have another induction scheme\footnote{\texttt{Bochner\_Integration\_Supplement.integrable\_simple\_function\_induct\_nn} in \cite{Keskin_A_Formalization_of_2023}}.\label{ind:simple_nn}

The statement of this induction scheme seems even more complicated, but in essence it is the same induction scheme as the previous one with the added assumption of non-negativity everywhere. The proof is also largely the same. We just need to show that the partial sums stay non-negative all the way through.

\pagebreak
\subsection{Bochner Integration on Linearly Ordered Banach Spaces}

When working with the real numbers, the following statement is easy to show.
\begin{center}
Let $f, g : \Omega \rightarrow \mathbb{R}$ be integrable and $f \ge g$ $\mu$-a.e., then $\int f\;\textrm{d} \mu \ge \int g\;\textrm{d} \mu$. 
\end{center}
In this subsection, we aim to provide similar results for functions $f, g : \Omega \rightarrow E$ with $E$ a linearly ordered Banach space. For the remainder of our discourse, a topological space $E$ is linearly ordered, if there exists a total ordering on $E$ such that the topology on $E$ and the order topology induced by the ordering coincide.

We start with the following lemma

\begin{lemma}\label{nonneg_integral}
	Let $f \in L^1(E)$ and $f \ge 0$ $\mu$-a.e. Then $\int f \;\textrm{d}\mu \ge 0$.
\end{lemma}
\begin{proof}
	Since $f \in L^1(E)$, there exists a sequence of integrable simple functions $(s_n)_{n \in \mathbb{N}}$, such that $\lim_{n \to \infty} s_n(x) = f(x)$ $\mu$-a.e. and $\lim_{n \to \infty} \int s_n \;\textrm{d}\mu = \int f \;\textrm{d}\mu$. At first, we have no further information about $s_n$. However, since we know that $f \ge 0$ $\mu$-a.e, it follows that $f = \max(0,f)$ $\mu$-a.e. Using dominated convergence and the fact that the function $\max(0,\cdot)$ is continuous w.r.t to the order topology on $E$, we can show
\[
	\lim_{n \to \infty} \max(0, s_n(x)) = \max(0, f(x)) \; \mu\textrm{-a.e.}
\]
and
\[
	\lim_{n \to \infty} \int \max(0, s_n) \;\textrm{d}\mu = \int \max(0, f) \;\textrm{d}\mu
\]
The function $\max(0, s_n)$ is still a simple and integrable function, which has the additional property of being always non-negative. 

We will now show that if $h$ is a non-negative simple function, then $\int h \;\textrm{d}\mu \ge 0$. For this purpose, we will use the induction scheme for non-negative integrable simple functions that we introduced in the previous subsection (\ref{ind:simple_nn}).

\paragraph{Congruence:} Let $h = g$ $\mu$-a.e. and $\int g \;\textrm{d}\mu \ge 0$ for some $g \in L^1(E)$. It follows directly
\[
	\int h \;\textrm{d}\mu = \int g \;\textrm{d}\mu \ge 0
\]

\paragraph{Indicator Functions:} Let $h = \mathbf{1}_A \cdot y$ for some measurable set $A$ with finite measure and $y \in E$ with $y \ge 0$. It follows directly 
\[
	\int h \;\textrm{d}\mu = \mu(A) \cdot y \ge 0
\]

\paragraph{Addition:} Let $h = h_1 + h_2$ for some integrable simple functions $h_1$ and $h_2$. By the induction hypothesis, we have $\int h_i \;\textrm{d}\mu \ge 0$ for $i = 1,2$. Therefore
\[
	\int h \;\textrm{d}\mu = \int h_1 \;\textrm{d}\mu + \int h_2 \;\textrm{d}\mu \ge 0
\]

Hence the statement holds for simple functions and therefore we know $\int \max(0, s_n) \;\textrm{d}\mu \ge 0$ for all $n \in \mathbb{N}$. Therefore, the same must hold for the limit $\lim_{n \to \infty} \int \max(0, s_n) \;\textrm{d}\mu = \int \max(0, f) \;\textrm{d}\mu$. Since $f = \max(0,f)$ $\mu$-a.e., we have $\int f \;\textrm{d}\mu = \int \max(0, f) \;\textrm{d}\mu$ and the statement follows.

\end{proof}
\begin{remark}
For the proof of this statement, we need the topology on $E$ to coincide with the order topology. Otherwise, we can not guarantee statements such as 
\[
	\forall n. \; x_n \ge 0 \implies \lim_{i \to \infty} x_i \ge 0
\]
or the continuity of the $\max$ function.
\end{remark}
This lemma entails the following corollary.
\begin{corollary}
	Let $f, g \in L^1(E)$ and $f \ge g$ $\mu$-a.e. Then $\int f\;\textrm{d} \mu \ge \int g\;\textrm{d} \mu$.
\end{corollary}

In Isabelle, we can replace the assumption $f \in L^1(E)$ with Borel measurability, since a non-integrable function has the value of its integral set to $0$ by default. Lemma \ref{nonneg_integral} can be stated as
{\small
\begin{lstlisting}[style=isabelle]
lemma integral_nonneg_AE_banach:
  assumes "$f \in \texttt{borel\_measurable} \; M$" and "AE$ \; x \; $in$ \; M. \; 0 \le f \; x$"
  shows "$0 \le \texttt{integral}^L \; M \; f$"
\end{lstlisting}
}

\section{Constructing the Conditional Expectation}

Before we can talk about \textit{the} conditional expectation, we must define what it means for a function to have \textit{a} conditional expectation. 

\begin{definition}
	Let $f \in L^1(E)$. Given a sub-$\sigma$-algebra $F \subseteq \Sigma$, we call an $F$-measurable function $g : \Omega \rightarrow E$ a conditional expectation of $f$ with respect to the sub-$\sigma$-algebra $F$, if the following equality holds for all $A \in F$

	\[
		\int_A f \; \textrm{d} \mu = \int_A g \; \textrm{d} \mu
	\]

\end{definition}

We formalize this notion by introducing the following predicate in the theory file \texttt{Conditional\_Expectation\_Banach} 

\begin{isadefinition}
{\small
\begin{lstlisting}[style=isabelle]
	definition has_cond_exp where 
	  "has_cond_exp $M \; F \; f \; g \; = (\forall A \in \texttt{sets} \; F. \; \int_A \; f \; \partial M = \int_A \; g \; \partial M)$
							  $\wedge \; \texttt{integrable} \; M \; f $
							  $\wedge \; \texttt{integrable} \; M \; g $
							  $\wedge \; g \in \texttt{borel\_measurable} \; F$"
\end{lstlisting}
}
\end{isadefinition}

This predicate precisely characterizes what it means for a function $f$ to have a conditional expectation $g$ with respect to the measure $M$ and the sub-$\sigma$-algebra $F$. Now we can use Hilbert's $\epsilon$-operator, \lstinline[language=isabelle]{SOME} in Isabelle \cite{Nipkow-Paulson-Wenzel:2002}, to define \textit{the} conditional expectation, if it exists. The use of the word \textit{the} will be fully justified after we show uniqueness in the next subsection.

\begin{isadefinition}
{\small
	\begin{lstlisting}[style=isabelle]
		definition cond_exp where 
		  "cond_exp $M \; F \; f$ 
			$= ($if$ \; \exists g. \; \texttt{has\_cond\_exp} \; M \; F \; f \; g \; $then$ \; ($SOME$ \; g. \; \texttt{has\_cond\_exp} \; M \; F \; f \; g) \; $else$ \; (\lambda \_. \; 0))$"
	\end{lstlisting}
}
\end{isadefinition}

A major advantage of defining the conditional expectation this way is that it allows us to make statements about its measurability and integrability, without needing to show existence or uniqueness. The following formal lemmas reflect this.

\begin{isalemma}
{\small
	\begin{lstlisting}[style=isabelle]
	lemma borel_measurable_cond_exp: "cond_exp $M \; F \; f \; \in$ borel_measurable $F$"
	\end{lstlisting}
}
\end{isalemma}

\begin{isalemma}
{\small
	\begin{lstlisting}[style=isabelle]
	lemma integrable_cond_exp: "integrable $M$ (cond_exp $M \; F \; f$)"
	\end{lstlisting}
}
\end{isalemma}

\subsection{Uniqueness}

The conditional expectation of a function is unique up to a $\mu$-null set. 

\begin{lemma}
	Let $f, g \in L^1(E)$ such that ``$\normalfont\texttt{has\_cond\_exp} \; M \; F \; f \; g$'' holds. Then 
	\par\noindent ``$\normalfont\texttt{has\_cond\_exp} \; M \; F \; f \; (\texttt{cond\_exp} \; M \; F \; f)$'' and
	\[
		\normalfont\texttt{cond\_exp} \; M \; F \; f = g \;\; \mu\texttt{-a.e.}
	\]
\end{lemma}
\begin{proof}
	The first statement follows directly from the definition of \texttt{cond\_exp}. To show ``$\texttt{cond\_exp} \; M \; F \; f = g$'' $\mu$-a.e. we argue as follows. By the definition of \texttt{has\_cond\_exp} we have for any $A \in F$
	\[
		\int_A \; f \; \textrm{d}\mu = \int_A \; g \; \textrm{d}\mu
	\]
	and
	\[
		\int_A \; f \; \textrm{d}\mu = \int_A \; \texttt{cond\_exp} \; M \; F \; f \; \textrm{d}\mu
	\]
	Together with the lemma on the uniqueness of densities, we have 
	\par\noindent``$\texttt{cond\_exp} \; M \; F \; f = g$'' $\mu\vert_F$-a.e. The lemma follows from the fact that all $\mu\vert_F$-null sets are also $\mu$-null sets.
\end{proof}

Hence, the defining property of the conditional expectation guarantees that it is unique up to a $\mu$-null set. We formalize this statement as follows.

\begin{isalemma}
{\small
	\begin{lstlisting}[style=isabelle]
	lemma cond_exp_charact:
	  assumes "$\bigwedge A \in \texttt{sets} \; F. \; \int_A \; f \; \partial M = \int_A \; g \; \partial M$"
			  "$\texttt{integrable} \; M \; f$"
			  "$\texttt{integrable} \; M \; g$"
			  "$g \in \texttt{borel\_measurable} \; F$"
	  shows "AE$ \; x \; $in$ \; M. \; \texttt{cond\_exp} \; M \; F \; f \; x = g \; x$"
	\end{lstlisting}
}
\end{isalemma}

\subsection{Existence}

Showing the existence is a bit more involved. Specifically, what we aim to show is that ``$\normalfont\texttt{has\_cond\_exp} \; M \; F \; f \; (\texttt{cond\_exp} \; M \; F \; f)$'' holds for any Bochner-integrable $f$. We will employ the standard machinery of measure theory. First, we will prove existence for indicator functions. Then we will extend our proof by linearity to simple functions. Finally we use a limiting argument to show that the conditional expectation exists for all Bochner-integrable functions.

The conditional expectation operator has already been formalized for real-valued functions by S\`ebastien Gou\"ezel via the definition \texttt{real\_cond\_exp}. The following lemmas show that our formal definition of the conditional expectation coincides with the existing definition in the real case.

\begin{isalemma}
{\small
	\begin{lstlisting}[style=isabelle]
lemma has_cond_exp_real:
  assumes "$\texttt{integrable} \; M \; f$"
  shows "$\texttt{has\_cond\_exp} \; M \; F \; f \; (\texttt{real\_cond\_exp} \; M \; F \; f)$"
	\end{lstlisting}
}
\end{isalemma}

\begin{isalemma}
{\small
	\begin{lstlisting}[style=isabelle]
lemma cond_exp_real:
  assumes "$\texttt{integrable} \; M \; f$"
  shows "AE$ \; x \; $in$ \; M. \; \texttt{cond\_exp} \; M \; F \; f \; x = \texttt{real\_cond\_exp} \; M \; F \; f \; x$" 
	\end{lstlisting}
}
\end{isalemma}

We can now show that the conditional expectation of indicator functions exist.

\begin{lemma}
	Let $A \subseteq \Omega$ be measurable with $\mu(A) < \infty$ and $y \in E$. Then
	\[
		\normalfont\texttt{has\_cond\_exp} \; M \; F \; (\mathbf{1}_A \cdot y) \; ((\texttt{real\_cond\_exp} \; M \; F \; \mathbf{1}_A) \cdot y)
	\]
\end{lemma}
\begin{proof}
	The statement follows directly from the linearity of the Bochner-integral and the previous lemmas.
\end{proof}

Next, we show the following lemma concerning the sum of two conditional expectations.

\begin{lemma}
	Assume $\normalfont\texttt{has\_cond\_exp} \; M \; F \; f \; f'$ and $\normalfont\texttt{has\_cond\_exp} \; M \; F \; g \; g'$. Then
	\[
		\normalfont\texttt{has\_cond\_exp} \; M \; F \; (f + g) \; (f' + g')
	\]
\end{lemma}
\begin{proof}
	The statement follows directly from the linearity of the Bochner-integral.
\end{proof}

Together with the induction scheme \texttt{integrable\_simple\_function\_induct}, we can show that the conditional expectation of an integrable simple function exists.

Now comes the most difficult part. Given a convergent sequence of integrable simple functions $(s_n)_{n \in \mathbb{N}}$, we must show that the sequence $(\texttt{cond\_exp} \; M \; F \; s_n)_{n \in \mathbb{N}}$ is also convergent. Furthermore, we must show that this limit satisfies the properties of a conditional expectation. Unfortunately, we will only be able to show that this sequence convergences in the $L^1$-norm. Luckily, this is enough to show that the operator $\texttt{cond\_exp} \; M \; F$ preserves limits as a function $L^1(E) \rightarrow L^1(E)$. We need the following lemma for this purpose

\begin{lemma} (Contractivity for Simple Functions) \par
	Let $f : \Omega \rightarrow E$ be an integrable simple function. Then
	\[
		\normalfont\lVert \texttt{cond\_exp} \; M \; F \; s \rVert \le \texttt{cond\_exp} \; M \; F \; (\lambda x. \lVert s \; x \rVert)
	\]
\end{lemma}
\begin{proof}
	In the real case, one can show this property by decomposing a function into positive and negative parts. The statement then follows via the induction scheme \par\noindent\texttt{integrable\_simple\_function\_induct}.
\end{proof}

The following lemma is the most involved result of our formalization.

\begin{lemma}
	Let $f : \Omega \rightarrow E$ be an integrable function. Let $(s_n)_{n \in \mathbb{N}}$ be a sequence of integrable simple functions, such that $\lim_{n \to \infty} s_n(x) = f(x)$ and $\forall n. \;\lVert s_n(x) \rVert \le 2 \cdot \lVert f(x) \rVert$ for $\mu$-almost all $x$. Then there exists some subsequence $(s_{r_n})_{n \in \mathbb{N}}$ such that
	\[
		\normalfont (\texttt{cond\_exp} \; M \; F \; s_{r_n})_{n \in \mathbb{N}} \; \textrm{is Cauchy} \; \mu\textrm{-a.e.}
	\]
	and
	\[
		\normalfont \texttt{has\_cond\_exp} \; M \; F \; f \; (\lim_{n \to \infty} \texttt{cond\_exp} \; M \; F \; s_{r_n})
	\]
\end{lemma}
\begin{proof}
	The sequence $(s_n)_{n \in \mathbb{N}}$ is Cauchy $\mu$-a.e. Hence $\lim_{n \to \infty} \textrm{diam}(S_n(x)) = 0$, with $S_n(x) := \{s_i(x) \; \vert \; i \ge n \}$ by the diameter lemma. Furthermore
	\[
		\lVert\textrm{diam}(S_n(x))\rVert \le 4 \cdot \lVert f(x) \rVert \; \mu\textrm{-a.e.}
	\]
	using the triangle inequality and our second assumption. We have already shown that $\textrm{diam}(S_n(x))$ is measurable. We apply the dominated convergence theorem and get
	\[
		\lim_{n \to \infty} \int \textrm{diam}(S_n(x)) \; \textrm{d}\mu = 0
	\]
	We will now show that $(\texttt{cond\_exp} \; M \; F \; s_n)_{n \in \mathbb{N}}$ is Cauchy in the $L^1$-norm.
	Let $\varepsilon > 0$. Hence there is some $N \in \mathbb{N}$ such that $\int \textrm{diam}(S_n(x)) < \varepsilon$. Thus for any $i,j \ge N$, we have
	\[
		\int \lVert s_i(x) - s_j(x) \rVert \; \textrm{d}\mu \le \int \textrm{diam}(S_N(x)) \; \textrm{d}\mu < \varepsilon
	\]
	by the monotonicity of the integral. Furthermore
	\begin{align*}
		&\quad\;\int \lVert (\texttt{cond\_exp} \; M \; F \; s_i)(x) - (\texttt{cond\_exp} \; M \; F \; s_j)(x) \rVert \; \textrm{d}\mu \\
		&= \int \lVert (\texttt{cond\_exp} \; M \; F \; (s_i - s_j))(x) \rVert \; \textrm{d}\mu \\
		&\le \int (\texttt{cond\_exp} \; M \; F \; (\lambda x.\; \lVert s_i(x) - s_j(x)\rVert))(x) \; \textrm{d}\mu \\
		&= \int \lVert s_i(x) - s_j(x)\rVert \; \textrm{d}\mu \\
		&< \varepsilon
	\end{align*}
	since $s_i(x)-s_j(x)$ is an integrable simple function and conditional expectation already exists in the real setting. Hence $(\texttt{cond\_exp} \; M \; F \; s_n)_{n \in \mathbb{N}}$ is Cauchy in the $L^1$-norm. Therefore, there exists some subsequence $(\texttt{cond\_exp} \; M \; F \; s_{r_n})_{n \in \mathbb{N}}$ that convergences $\mu$-a.e. We have for all $n \in \mathbb{N}$
	\[
		\lVert (\texttt{cond\_exp} \; M \; F \; s_{r_n}) (x) \rVert \le \texttt{cond\_exp} \; M \; F \; (\lambda x. \; 2 \cdot \lVert f(x)\rVert) \; \mu\textrm{-a.e.}
	\]
	Together with the dominated convergence theorem, this implies that \par\noindent$\lim_{n \to \infty} (\texttt{cond\_exp} \; M \; F \; s_{r_n})$ is integrable. \par\noindent As the limit of $F$-measurable functions, $\lim_{n \to \infty} (\texttt{cond\_exp} \; M \; F \; s_{r_n})$ is also $F$-measurable. Finally, we have for $A \in F$
	\begin{align*}
		\int_A (\lim_{n \to \infty} (\texttt{cond\_exp} \; M \; F \; s_{r_n}))(x) \; \textrm{d}\mu &= \lim_{n \to \infty} \int_A (\texttt{cond\_exp} \; M \; F \; s_{r_n})(x) \; \textrm{d}\mu \\
		&= \lim_{n \to \infty} \int_A s_{r_n}(x) \; \textrm{d}\mu \\
		&= \int_A f(x) \; \textrm{d}\mu 
	\end{align*}
	In the first and last equality we have again used the dominated convergence theorem. The statement follows from the definition of \texttt{has\_cond\_exp}.
\end{proof}

At one point in the proof of our lemma, we have used the fact that a convergent sequence in $L^1$ admits a subsequence which is convergent in the underlying norm $\mu$-a.e. This result is stated in Isabelle as follows

{\small
	\begin{lstlisting}[style=isabelle]
	proposition tendsto_L1_AE_subseq:
	  fixes u :: "$\texttt{nat} \Rightarrow 'a \Rightarrow 'b$"
	  assumes "$\bigwedge n. \; \texttt{integrable} \; M \; (u \; n)$"
		  and "$(\lambda n. \; (\int \; \texttt{norm} \;(u \; n \; x) \; \partial M)) \longrightarrow 0$"
	  shows "$\exists r \;:: \; \texttt{nat} \Rightarrow \texttt{nat}. \; \texttt{strict\_mono} \; r \wedge ($AE$\; x \; $in$ \; M. \; (\lambda n. \; u \; (r \; n) \; x) \longrightarrow 0)$"
	\end{lstlisting}
}

In our case, it is cumbersome to formulate the convergence of $(\texttt{cond\_exp} \; M \; F \; s_n)_{n \in \mathbb{N}}$ in the $L^1$-norm in the manner stated above. One might be tempted to use the diameter lemma in the other direction to bring the expression into this form. On paper this is indeed plausible. However, it involves showing that the functions $\texttt{cond\_exp} \; M \; F \; s_n$ have an integrable upper bound $w$. Furthermore, one has to jump back and forth between using the binder \lstinline[mathescape]{AE $x$ in $M$.} and directly showing statements for $x \in \texttt{space} \; M$. We have decided that it would be easier to formalize and use the following more flexible lemma instead. Mathematically, the underlying argument is the same.

\begin{isalemma}
{\small
	\begin{lstlisting}[style=isabelle]
	lemma cauchy_L1_AE_cauchy_subseq:
	  fixes s :: "$\texttt{nat} \Rightarrow 'a \Rightarrow 'b$"
	  assumes "$\bigwedge n. \; \texttt{integrable} \; M \; (s \; n)$"
		  and "$\bigwedge e. \; e > 0 \implies \exists N. \; \forall i\ge N. \forall j \ge N. \; $LINT$ \; x\vert M. \; \texttt{norm} \; (s \; i \; x \; - \; s \; j \; x) < e$"
	  obtains $r$ where "$\texttt{strict\_mono} \; r$" "AE$ \; x \; $in$ \; M. \; \texttt{Cauchy} \; (\lambda i. \; s \; (r \; i) \; x)$"
	\end{lstlisting}
}
\end{isalemma}

The main result of this subsection is formalized in Isabelle as follows

\begin{isacorollary}
{\small
	\begin{lstlisting}[style=isabelle]
	corollary has_cond_expI:
	  assumes "$\texttt{integrable} \; M \; f$"
	  shows "$\texttt{has\_cond\_exp} \; M \; F \; f \; (\texttt{cond\_exp} \; M \; F \; f)$"
	\end{lstlisting}
}
\end{isacorollary}

\subsection{Properties of the Conditional Expectation}

We will now introduce some commonly used properties of the conditional expectation. The proofs for the last two properties presented in this subsection are derived from the lecture notes by Gordan Zitkovic for the course ``Theory of Probability I'' \cite{Zitkovic_2015}.

\subsubsection{Identity on $F$-measurable functions}

If an integrable function $f$ is already $F$-measurable, then ``$\texttt{cond\_exp} \; M \; F \; f = f$'' $\mu$-a.e. This is a corollary of the lemma on the characterization of \texttt{cond\_exp}. It is formalized in our development as follows.
\begin{isacorollary}
{\small
	\begin{lstlisting}[style=isabelle]
	corollary cond_exp_F_meas:
	  assumes "$\texttt{integrable} \; M \; f$"
			  "$f \in \texttt{borel\_measurable} \; F$"
	  shows "AE$ \; x \; $in$ \; M. \; \texttt{cond\_exp} \; M \; F \; f \; x = f \; x$"
	\end{lstlisting}
}
\end{isacorollary}

\subsubsection{Tower Property}

The following property is called the \textit{tower property} of the conditional expectation.

\begin{lemma}
	Let $F$ and $G$ be nested sub-$\sigma$-algebras, i.e. $F \subseteq G \subseteq \Sigma$. Then, for any $f \in L^1(E)$, we have
	\[
		\normalfont \texttt{cond\_exp} \; M \; F \; (\texttt{cond\_exp} \; M \; G \; f) = \texttt{cond\_exp} \; M \; F \; f \; \mu\textrm{-a.e.}
	\]
\end{lemma}
\begin{proof}
	For any $A \in F$, we have
	\begin{align*}
		\int_A \texttt{cond\_exp} \; M \; G \; f \; \textrm{d} \mu &= \int_A f \; \textrm{d} \mu \\
		&= \int_A \texttt{cond\_exp} \; M \; F \; f \; \textrm{d} \mu
	\end{align*}
	since $A$ is also in $G$. The characterization lemma yields the result.
\end{proof}

\subsubsection{Contractivity}

A linear operator $L : V \rightarrow W$ between normed vector spaces $V$ and $W$ is called a \textit{contraction} if its operator norm
\[
	\lVert L \rVert_{\texttt{op}} = \inf\{ c \geq 0 \;\vert\; \lVert Lv \rVert_W \leq c \lVert v \rVert_V \textrm{ for all } v \in V \}
\]
is less than or equal to 1. Such an operator always preserves limits and has other useful properties in functional analysis \cite{sznagy2010}.

\begin{lemma} (Contractivity) \par
	Let $f \in L^1(E)$. Then
	\[
		\normalfont\lVert \texttt{cond\_exp} \; M \; F \; f \rVert \le \texttt{cond\_exp} \; M \; F \; (\lambda x. \lVert f(x) \rVert)
	\]
\end{lemma}
\begin{proof}
	We have already shown contractivity in the case of simple functions. Since $f$ is integrable, there exists a sequence of simple functions $(s_n)_{n \in \mathbb{N}}$ such that
	\[
		\lim_{n \rightarrow \infty} s_n= f \; \mu\textrm{-a.e.}
	\]
	and
	\[
		\lVert s_n(x) \rVert \le 2 \cdot \lVert f(x) \rVert \; \mu\textrm{-a.e. for all } n \in \mathbb{N}
	\]
	Using the results of the previous subsection, we obtain a subsequence $(s_{r_n})_{n \in \mathbb{N}}$ such that
	\[
		\lim_{n \rightarrow \infty} (\texttt{cond\_exp} \; M \; F \; s_{r_n}) = \texttt{cond\_exp} \; M \; F \; f \quad \mu\textrm{-a.e.}
	\]
	With the exact same arguments applied to the sequence of simple functions $(\lambda x. \lVert s_{r_n}(x) \rVert)_{n \in \mathbb{N}}$, we obtain a sub-subsequence $(s_{r_{r'_n}})_{n \in \mathbb{N}}$ such that
	\[
		\lim_{n \rightarrow \infty} (\texttt{cond\_exp} \; M \; F \; (\lambda x. \lVert s_{r_{r'_n}}(x) \rVert)) = \texttt{cond\_exp} \; M \; F \; (\lambda x. \lVert f(x)\rVert) \quad\mu\textrm{-a.e.}
	\]
	Furthermore, we have
	\[
		\lVert (\texttt{cond\_exp} \; M \; F \; s_{r_{r'_n}})(x) \rVert \le (\texttt{cond\_exp} \; M \; F \; (\lambda x. \lVert s_{r_{r'_n}}(x) \rVert))(x) \quad\mu\textrm{-a.e.}
	\]
	for all $n \in \mathbb{N}$, since the functions in question are simple. Taking the limits on both sides and using the continuity of the norm yields the result.
\end{proof}

\begin{corollary}
	The linear operator $\normalfont\texttt{cond\_exp} \; M \; F : L^1(E) \rightarrow L^1(E)$ is a contraction.
\end{corollary}
\begin{proof}
	Let $f \in L^1(E)$. From the previous lemma we have 
	\begin{align*}
		\lVert \texttt{cond\_exp} \; M \; F \; f \rVert_1 &= \int \lVert \texttt{cond\_exp} \; M \; F \; f \rVert \; \textrm{d} \mu \\
		&\le \int \texttt{cond\_exp} \; M \; F \; (\lambda x. \lVert f(x) \rVert) \; \textrm{d} \mu \\
		&= \int \lVert f \rVert \textrm{d} \mu = \lVert f \rVert_1\\
	\end{align*}
	Hence $\lVert \texttt{cond\_exp} \; M \; F \rVert_{\texttt{op}} \le 1$
\end{proof}

\subsubsection{Pulling Out What's Known}

The following property of the conditional expectation is called ``pulling out what's known''. 

\begin{lemma}
	Let $f : \Omega \rightarrow \mathbb{R}$ be an $F$-measurable function. Let $g \in L^1(E)$ and $f \cdot g \in L^1(E)$. Then
	\[
		\normalfont \texttt{cond\_exp} \; M \; F \; (f \cdot g) = f \cdot \texttt{cond\_exp} \; M \; F \; g \quad \mu\textrm{-a.e.}
	\]
\end{lemma}
\begin{proof}
	The proof of this lemma is involved as well. Therefore we will only focus on the core idea of the proof. We will also assume that the result already holds in the real setting. We show the following seemingly less general statement for $z : \Omega \rightarrow \mathbb{R}$ $F$-measurable and $z \cdot g \in L^1(E)$:
	\[
		\int z \cdot g \; \textrm{d} \mu = \int z \cdot \texttt{cond\_exp} \; M \; F \; g \; \textrm{d}
	\]
	The result will follow by taking $z = f \cdot \mathbf{1}_A$ for $A \in F$.
	Since $z$ is measurable, there exists some sequence of simple functions $(s_n)_{n \in \mathbb{N}}$ such that
	\[
		\lim_{n \rightarrow \infty} s_n= z \; \mu\textrm{-a.e.}
	\]
	and
	\[
		\lvert s_n(x) \rvert \le 2 \cdot \lvert z(x) \rvert \; \mu\textrm{-a.e. for all } n \in \mathbb{N}
	\]
	In this case one can easily check that
	\[
		\int s_n \cdot g \; \textrm{d} \mu = \int s_n \cdot \texttt{cond\_exp} \; M \; F \; g \; \textrm{d}
	\]
	for all $n \in \mathbb{N}$
	
	By our additional assumption that the result already holds in the real case, we have
	\[
	\lvert z \cdot \texttt{cond\_exp} \; M \; F \; (\lambda x. \lVert g(x \rVert))\rvert = \texttt{cond\_exp} \; M \; F \; (\lambda x. \lvert z(x) \cdot \lVert g(x)\rVert\rvert)
	\]
	Using the contractivity of the conditional expectation and the above bound on $s_n$, it follows that 
	\[
		\lVert s_n \cdot \texttt{cond\_exp} \; M \; F \; g \rVert \le 2 \cdot \texttt{cond\_exp} \; M \; F \; (\lambda x. \lvert z(x) \cdot \lVert g(x)\rVert\rvert)
	\]
	Applying the dominated convergence theorem twice, we get
	\[
		\lim_{n \rightarrow \infty} \int s_n \cdot g \; \textrm{d} \mu = \int z \cdot g \; \textrm{d} \mu
	\]
	and
	\[
		\lim_{n \rightarrow \infty} \int s_n \cdot \texttt{cond\_exp} \; M \; F \; g \; \textrm{d} \mu = \int z \cdot \texttt{cond\_exp} \; M \; F \; g \; \textrm{d} \mu
	\]
	Since the sequence on the right hand side are equal, the statement follows from the fact that limits are unique.
\end{proof}

\subsubsection{Irrelevance of Independent Information}

Let $M = (\Omega, \Sigma, \mu)$ be a probability space, i.e. $\mu(\Omega) = 1$. Two sub-$\sigma$-algebras $F \subseteq \Sigma$ and $G \subseteq \Sigma$ are called independent, if for any $A \in F$ and $B \in G$ we have
\[
	\mu(A \cap B) = \mu(A) \mu(B)
\]
\begin{lemma}
Let $(E, \lVert \cdot\rVert)$ be a complete and normed field. Let $f : \Omega \rightarrow E$ be an integrable random variable. Let $G$ be a sub-$\sigma$-algebra which is independent from $\sigma(F \cup \sigma(f))$, where $\sigma(f)$ is the $\sigma$-algebra generated by $f$. Then,
\[
	\normalfont \texttt{cond\_exp} \; M \; (\sigma(F \cup G)) \; f = \texttt{cond\_exp} \; M \; F \; f \quad \mu\textrm{-a.e.}
\]
In particular if $f$ is independent of $F$, then 
\[
	\normalfont \texttt{cond\_exp} \; M \; F \; f = \left(\int_\Omega f \; \textrm{d} \mu\right) \quad \mu\textrm{-a.e.}
\]
\end{lemma}
\begin{proof}
	We need to show that
	\[
		\int \mathbf{1}_A \cdot f \; \textrm{d} \mu = \int \mathbf{1}_A \cdot \texttt{cond\_exp} \; M \; F \; f \; \textrm{d} \mu
	\]
	for all $A \in \sigma(F \cup G)$. Let $D$ be the collection of all $A \in \sigma(F \cup G)$ such that the statement above holds. It is clear that $D$ is a Dynkin system (also called a $\lambda$-system). The set $\{A \cap B \;\vert\; A \in F \wedge B \in G\}$ generates the $\sigma$-algebra $\sigma(F \cup G)$. Hence, it is enough to show that the property holds for elements of this set. Let $A \in F$ and $B \in G$. We have
\begin{align*}
	\int \mathbf{1}_{A \cap B} \cdot f \; \textrm{d} \mu &= \int \mathbf{1}_A \mathbf{1}_B \cdot f \; \textrm{d} \mu \\
	&= \left(\int \mathbf{1}_A \cdot f \; \textrm{d} \mu\right) \left(\int \mathbf{1}_B\; \textrm{d} \mu \right) \\
	&= \left(\int \mathbf{1}_A \cdot \texttt{cond\_exp} \; M \; F \; f \; \textrm{d} \mu \right) \left(\int \mathbf{1}_B\; \textrm{d} \mu \right) \\
	&= \int \mathbf{1}_{A \cap B} \cdot \texttt{cond\_exp} \; M \; F \; f \; \textrm{d} \mu
\end{align*}
Here we have used the independence of $\mathbf{1}_A \cdot f$ and $\mathbf{1}_B$, as well as the independence of $\mathbf{1}_A \cdot \texttt{cond\_exp} \; M \; F \; f$ and $\mathbf{1}_B$
\end{proof}

In order to show this property in Isabelle, we have used the induction scheme \texttt{sigma\_sets\_induct\_disjoint}, which lets us replicate the $\lambda$-system argument used above.

\section{Conditional Expectation on Linearly Ordered Banach Spaces}

In the presence of a linear ordering, we can prove certain monotonicity properties of the conditional expectation. We start with the following two lemmas

\begin{lemma}
	Let $f \in L^1(E)$. Assume $f \ge c$ $\mu$-a.e. for some $c \in E$. Then 
	\[
		\normalfont\texttt{cond\_exp} \; M \; F \; f \ge c \quad \mu\textrm{-a.e.}
	\]
\end{lemma}
\begin{proof}
	We will show the statement using the averaging theorem (\ref{thm:averaging_theorem}). Let $A \in F$ be a measurable set with $\mu(A) < \infty$. Then
	\begin{align*}
		c &= \frac{1}{\mu(A)} \int_A c \; \textrm{d} \mu \\
		&\le \frac{1}{\mu(A)} \int_A f \; \textrm{d} \mu \\
		&= \frac{1}{\mu(A)} \int_A \texttt{cond\_exp} \; M \; F \; f \; \textrm{d} \mu \\
		&= \frac{1}{\mu(A)} \int_A \texttt{cond\_exp} \; M \; F \; f \; \textrm{d} \mu\vert_F \\
	\end{align*}
	Hence $\int_A \texttt{cond\_exp} \; M \; F \; f \; \textrm{d} \mu\vert_F \in \{x \in E \;\vert\; x \ge c \}$. The statement follows from the fact that $\{x \in E \;\vert\; x \ge c \}$ is closed.
\end{proof}

\begin{lemma}
	Let $f \in L^1(E)$. Assume $f > c$ $\mu$-a.e. for some $c \in E$. Then 
	\[
		\normalfont\texttt{cond\_exp} \; M \; F \; f > c \quad \mu\textrm{-a.e.}
	\]
\end{lemma}
\begin{proof}
	The averaging theorem is not applicable in this case since $\{x \in E \;\vert\; x > c \}$ is not closed. Therefore, we argue as follows.
	
	Let $S = \{\texttt{cond\_exp} \; M \; F \; f \le c\}$. The conditional expectation $\texttt{cond\_exp} \; M \; F \; f$ is $F$-measurable, hence $S \in F$. Since $F$ is a $\sigma$-finite sub-$\sigma$-algebra, we can assume without loss of generality that $\mu(S) < \infty$.
	The assumption $f > c$ $\mu$-a.e. implies
	\[
		\int_S f \; \textrm{d} \mu \ge \int_S c \; \textrm{d} \mu
	\]
	Furthermore, by the definition of $S$
	\begin{align*}
		\int_S c \; \textrm{d} \mu &\ge \int_S \texttt{cond\_exp} \; M \; F \; f \; \textrm{d} \mu \\
		&= \int_S f \; \textrm{d} \mu
	\end{align*}
	Hence $\int_S f \; \textrm{d} \mu = \int_S c \; \textrm{d} \mu$. By Corollary 3.1.16, we have
	\[
		\mathbf{1}_S \cdot f = \mathbf{1}_S \cdot c \quad \mu\textrm{-a.e.}
	\]
	Because of our assumption $f > c$ $\mu$-a.e., this can only be the case if $S$ is a $\mu$-null set, which completes the proof.
\end{proof}

The corresponding lemmas for $(\le)$ and $(<)$ are simple corollaries. The operator's monotonicity is also a corollary.

\begin{corollary}
	Let $f, g \in L^1(E)$. Assume $f \ge g$ $\mu$-a.e. Then 
	\[
		\normalfont\texttt{cond\_exp} \; M \; F \; f \ge \texttt{cond\_exp} \; M \; F \; g \quad \mu\textrm{-a.e.}
	\]
\end{corollary}
\begin{proof}
	From the assumptions, we have $f - g \ge 0$ $\mu$-a.e. and $f - g \in L^1(E)$. Hence
	\[
		\texttt{cond\_exp} \; M \; F \; (f - g) = \texttt{cond\_exp} \; M \; F \; f - \texttt{cond\_exp} \; M \; F \; g \ge 0
	\]
\end{proof}

Apart from some auxillary lemmas, this wraps up our overview of the formalization of the conditional expectation operator.

%% file: 04_stochastic_processes.tex

\chapter{Stochastic Processes}\label{chapter:stochastic_processes}

It would not make sense to talk about martingales without introducing stochastic processes first. In standard terminology, a stochastic process is a collection of random variables defined on the same probability space. The indexing set often represents time, and each random variable in the collection corresponds to an outcome at a specific point of time in that set.

Take the example of stock price movement, where each day's stock price is a random variable influenced by a variety of uncertain factors. This sequence of prices forms a stochastic process, describing the stock's behavior. Another instance is the Poisson process, which models events like customer arrivals at a service center. This process captures the randomness in the timing of arrivals, aiding in optimizing resource allocation and enhancing customer service. In physics, Brownian motion characterizes the unpredictable and continuous trajectory followed by particles suspended in a medium due to random collisions with surrounding molecules, which is again modelled as a stochastic process. The theory of stochastic processes is the cornerstone for analysing randomness and building models that mirror real-world uncertainties.

Keeping this in consideration, we aim to build a comprehensive foundation for a theory of stochastic process in Isabelle. Since the definition is so straightforward, it usually suffices to just consider a collection of measurable functions to make formal statements about stochastic processes. There is not much to gain from making an explicit definition on its own. Nonetheless, we must create a framework to discuss stochastic processes that can afterwards be broadened to formalize concepts like adaptedness and predictability. Locales present themselves as the solution we are looking for.

The locale system in Isabelle is useful for managing large formal developments, as it promotes modularity and reusability. It allows us to define generic theorems and structures in one place and then reuse them in multiple contexts without duplicating efforts. For instance, when defining filtered measure spaces in the following section, we will need to have an element act as the de facto bottom element of an index type. Locales allow us to easily fix such an element for this purpose. The results formalized in this chapter can be found in the theory files \texttt{Martingale.Filtered\_Measure}, \texttt{Martingale.Measure\_Space\_Supplement} and \texttt{Martingale.Stochastic\_Process} in \cite{Keskin_A_Formalization_of_2023}.
\vspace{2em}\\

We state the definition of a stochastic process.

\begin{definition}
	Let $(\Omega, \Sigma, \mu)$ be a measure space. Let $(X_i)_{i \in I}$ be a family of functions with $X_i : \Omega \rightarrow E$ for $i \in I$. The collection $(X_i)_{i \in I}$ is called a stochastic process if $X_i$ is $\Sigma$-measurable for all $i \in I$.
\end{definition}

Subsequently, we introduce the corresponding locale.

\begin{isadefinition}
{\small
\begin{lstlisting}[style=isabelle]
locale stochastic_process =
  fixes $M \; t_0$ 
    and $X$ :: "$'b \; :: \; \{$second_countable_topology, order_topology, t2_space$\} \Rightarrow \; 'a \; \Rightarrow \; 'c$"
  assumes random_variable[measurable]: "$\bigwedge i. \; t_0 \le i \implies X \; i \in \texttt{borel\_measurable} \; M$"
\end{lstlisting}
}
\end{isadefinition}

The measure $M$ represents the underlying measure space on which the stochastic process is defined. The index $t_0$ represents the initial point in time beyond which the process $X$ should be defined. As such, this locale formalizes a stochastic process defined on the set $[t_0, \infty) = \{t \;\vert\; t \ge t_0\}$. In the upcoming section, we will discuss why we chose to constrain the index set in this manner.

We have the following introduction lemmas for ``constant'' stochastic processes.

\begin{isalemma}
{\small
\begin{lstlisting}[style=isabelle]
lemma stochastic_process_const_fun:
  assumes "$f \in \texttt{borel\_measurable} \; M$"
  shows "$\texttt{stochastic\_process} \; M \; t_0 \; (\lambda \_. \; f)$"

lemma stochastic_process_const:
  shows "$\texttt{stochastic\_process} \; M \; t_0 \; (\lambda i \; \_. \; c \; i)$" 

\end{lstlisting}
}
\end{isalemma}

The lemmas state that we can define a stochastic process by setting $X_i = f$ for a measurable function $f : \Omega \rightarrow E$, or via $X_i = c$ for some $c \in E$. Furthermore, by composing a Borel-measurable function with a stochastic process, we can obtain another stochastic process. We formalize this statement as follows.

\begin{isalemma}
{\small
\begin{lstlisting}[style=isabelle]
lemma compose_stochastic:
  assumes "$\bigwedge i. \; t_0 \le i \implies f \; i \in \texttt{borel\_measurable} \; \texttt{borel}$"
  shows "$\texttt{stochastic\_process} \; M \; t_0 \; (\lambda i \; x. \; (f \; i) \; (X \; i \; x))$"
  \end{lstlisting}
}
\end{isalemma}

In the upcoming sections, we will observe how the assumptions for these statements change, when we place further restrictions on the collection of functions $(X_t)_{t \in [t_0,\infty)}$. For sake of completeness, we also provide lemmas which show that stochastic processes are stable under various operations on a vector space, such as norming, multiplication by a scalar valued function, addition, and partial sums over indices.

We also introduce the following sublocales to easily make statements about discrete-time and continuous-time stochastic processes.

\begin{isadefinition}
{\small
\begin{lstlisting}[style=isabelle]
locale nat_stochastic_process = stochastic_process $M$ "0 :: nat" $X$ for $M \; X$
locale real_stochastic_process = stochastic_process $M$ "0 :: real" $X$ for $M \; X$
\end{lstlisting}
}
\end{isadefinition}

By explicitly designating an element $t_0$ to be the bottom element, we can formalize continuous-time stochastic processes, i.e. $(X_t)_{t \in \mathbb{R}_{\ge 0}}$, without the need for introducing a new type for non-negative real numbers. 

\begin{remark}
Moving forward, we will define the concepts of adaptedness, progressive measurability and predictability. In our formalization, we have introduced analogous lemmas and sublocales for these process varieties as well. To avoid repeating ourselves, we will only reiterate these statements, if the proofs become non-trivial or if the assumptions change.
\end{remark}

Before presenting the remaining process varities, we must introduce the concept of a filtered measure space.

\section{Filtered Measure Spaces}

A filtered measure space is a measure space equipped with a sequence of increasing sub-$\sigma$-algebras, called a \textit{filtration} that represents the accumulation of information over time.

Let $M$ be a measure space. Assume we have a sequence of $\sigma$-algebras $(F_n)_{n \in \mathbb{N}}$ where 
\[
	F_0 \subseteq F_1 \subseteq F_2 \subseteq \dots
\]
This sequence forms a filtration on $M$. Intuitively, each $F_n$ represents the information available up to time $n$. In general, the index set does not need to be countable. The exact definition of a filtration is as follows. 
\begin{definition}
	Let $M = (\Omega, \Sigma, \mu)$ be a measure space and $(F_t)_{t \in I}$ a collection of sub-$\sigma$-algebras of $\Sigma$. $(F_t)_{t \in I}$ is called a \textit{filtration} of the measure space $M$, if for any $i, j\in I$ with $i \le j$ we have 
	\[
		F_i \subseteq F_j
	\]
	Hence, the sub-$\sigma$-algebras $F_i$ reflect the accumulation of information over time.
\end{definition}

In Isabelle, we define the following locale to capture this concept.

\begin{isadefinition}
{\small
\begin{lstlisting}[style=isabelle]
locale filtered_measure = 
  fixes $M \; F$ and $t_0$ :: "$'b \; :: \; \{$second_countable_topology, order_topology, t2_space$\}$"
  assumes subalgebra: "$\bigwedge i. \; t_0 \le i \implies \texttt{subalgebra} \; M \; (F \; i)$"
      and sets_F_mono: "$\bigwedge i \; j. \; t_0 \le i \implies i \le j \implies \texttt{sets} \; (F \; i) \subseteq \texttt{sets} \; (F \; j)$"
\end{lstlisting}

with the predicate \texttt{subalgebra} in \texttt{HOL-Probability.Conditional\_Expectation} defined via

\begin{lstlisting}[style=isabelle]
definition subalgebra where
  "$\texttt{subalgebra} \; M \; F = ((\texttt{space} \; F = \texttt{space} \; M) \wedge (\texttt{sets} \; F \subseteq \texttt{sets} \; M))$"
  \end{lstlisting}
}
\end{isadefinition}

\begin{remark}
	In Isabelle the $\texttt{measure}$ type is used to represent both measure spaces and $\sigma$-algebras. The latter is achieved by only considering the underlying $\sigma$-algebra via the projection ``\texttt{sets}''.
\end{remark}

In general, a type with an ordering does not necessarily inhabit a bottom element, i.e. an element that is less than or equal to any other element. In the next section, we will see how the existence of a bottom element lets us easily make statements concerning when a family of random variables constitutes an adapted process. From a practical point of view, this is not too much to assume, since all random processes one encounters in the real world must start at some fixed point in time (or at least that assumption can be made for practical purposes).

The keen reader might have noticed that we need a little bit more to define martingales properly. Namely, the sub-$\sigma$-algebras that comprise the filtration $(F_n)_{n \in \mathbb{N}}$ must all be $\sigma$-finite. Otherwise, we can not make use of our lemmas concerning the conditional expectation. We introduce the following locale to adress this issue.

\begin{isadefinition}
{\small
\begin{lstlisting}[style=isabelle]
locale sigma_finite_filtered_measure = filtered_measure +
  assumes sigma_finite: "sigma_finite_subalgebra $M \; (F \; t_0)$"
  \end{lstlisting}
}
\end{isadefinition}

\begin{remark}
	Since we artifically designated an element $t_0$ to represent the least index in consideration, we only need to show $\sigma$-finiteness for the sub-$\sigma$-algebra $F_{t_0}$. $\sigma$-finiteness of all other sub-$\sigma$-algebras follows from the monotonicity of the filtration.
\end{remark}

For the sake of completeness, we also introduce a local covering the case where the measure space is finite.

\begin{isadefinition}
{\small
\begin{lstlisting}[style=isabelle]
locale finite_filtered_measure = filtered_measure + finite_measure
  \end{lstlisting}
}
\end{isadefinition}

In order to make the ideas in this section a bit more concrete, we will now discuss two example filtrations. 

If we have some sub-$\sigma$-algebra $F \subseteq \Sigma$, then we can trivially take as our filtration $F_i = F$ for all $i \in [t_0,\infty)$. If we additionally know that we are working with a $\sigma$-finite subalgebra, then this yields a trivial $\sigma$-finite filtration on $M$. This choice of filtration is called a \textbf{constant filtration}. In Isabelle, we have the following lemma and the sublocale relation to reflect this.

\begin{isalemma}
{\small
\begin{lstlisting}[style=isabelle]
lemma filtered_measure_constant_filtration:
  assumes "$\texttt{subalgebra} \; M \; F$"
  shows "$\texttt{filtered\_measure} \; M \; (\lambda \_. \; F) \; t_0$"

sublocale sigma_finite_subalgebra $\subseteq$ constant_filtration: 
		sigma_finite_filtered_measure $M$ "$(\lambda$_. $F$)" $t_0$
\end{lstlisting}
}
\end{isalemma}

\begin{remark}
	Both of the statements above convey the same information essentially. The first one is stated in terms of premises and results, the latter in the language of locales. The notion of a $\sigma$-algebra being a subalgebra is formalized via the predicate \texttt{subalgebra}. Had the formalization been done in the language of locales, we could replace the first statement with an equivalent sublocale relation.
\end{remark}

Preparing for our next example, we introduce a formalization for the notion of a $\sigma$-algebra generated by a family of functions:

\begin{isadefinition}
{\small
\begin{lstlisting}[style=isabelle]
definition family_vimage_algebra where
  "$\texttt{family\_vimage\_algebra} \; \Omega \; S \; N  \equiv \texttt{sigma} \; \Omega \; (\bigcup f \in S. \; \{(f \;\text{-\textasciigrave} \; A) \; \cap \; \Omega \;\vert\; A. \; A \in N\})$"
\end{lstlisting}
}
\end{isadefinition}

Given two measure spaces $(V, \mathcal{A})$ and $(W, \mathcal{B})$, it is a well known fact that a function $f : V \rightarrow W$ is measurable, if and only if the generated $\sigma$-algebra $\sigma(f)$ is a subalgebra of $\mathcal{A}$. This result is captured for families of functions in the following lemma.

\begin{isalemma}
{\small
\begin{lstlisting}[style=isabelle]
lemma measurable_family_iff_sets:
  shows "$(S \subseteq N \rightarrow_M M) \longleftrightarrow S \subseteq (\texttt{space}\; N \rightarrow \;\texttt{space}\; M) \; \wedge$
                            $\texttt{family\_vimage\_algebra} \; (\texttt{space}\; N) \; S\; M \subseteq N$"
\end{lstlisting}
}
\end{isalemma}

Now, we can introduce our more interesting example, the \textbf{natural filtration}. Formally, we define it as follows.

\begin{isadefinition}
{\small
\begin{lstlisting}[style=isabelle]
definition natural_filtration where
  "$\texttt{natural\_filtration} \; M \; t_0 \; Y$
	$= (\lambda t. \; \texttt{family\_vimage\_algebra} \; (\texttt{space} \; M) \; \{Y \; i \; \vert\; i. \; i \in \{t_0..t\}\}\; \texttt{borel})$"
\end{lstlisting}
}
\end{isadefinition}
The natural filtration with respect to a stochastic process $Y$ is the filtration generated by all events involving the process up to the time index $t$, i.e. $F_t = \sigma(\{Y_i \; \vert\; i. \; i \le t\})$. Assuming that $Y$ is a stochastic process, i.e. $Y_i$ is $\Sigma$-measurable for all $i \ge t_0$, the definition indeed provides a filtration. The following sublocale relation formalizes this.

\begin{isalemma}
{\small
\begin{lstlisting}[style=isabelle]
sublocale stochastic_process $\subseteq$ filtered_measure_natural_filtration: 
	filtered_measure $M$ "natural_filtration $M$ $t_0$ $X$" $t_0$
\end{lstlisting}
}
\end{isalemma}

The natural filtration contains information concerning the process's past behavior at each point in time. The natural filtration is essentially the simplest filtration for studying a process. However, the natural filtration is not always $\sigma$-finite. In order to show that the natural filtration gives rise to a $\sigma$-finite filtered measure, we need to provide a countable exhausting set in the preimage of $X_{t_0}$. This statements is also present in our formalization\footnote{\texttt{Stochastic\_Process.sigma\_finite\_filtered\_measure\_natural\_filtration} in \cite{Keskin_A_Formalization_of_2023}}.

Of course, if the measure is already finite, the filtered measure space is also finite:

\begin{isalemma}
{\small
\begin{lstlisting}[style=isabelle]
lemma (in finite_measure) finite_filtered_measure_natural_filtration:
  assumes "$\texttt{stochastic\_process} \; M \; t_0 \; X$"
  shows "$\texttt{finite\_filtered\_measure} \; M \; (\texttt{natural\_filtration} \; M \; t_0 \; X) \; t_0$"
\end{lstlisting}
}
\end{isalemma}

This concludes our development of filtered measure spaces.

\section{Adapted Processes}

\begin{definition}

Let $(F_t)_{t \in [t_0, \infty)}$ be a filtration of the measure space $M = (\Omega, \Sigma, \mu)$. A stochastic process $(X_t)_{t \in [t_0, \infty)}$ is called \textit{an adapted process} if, for every index $t \ge t_0$, the random variable $X_t$ is measurable with respect to the $\sigma$-algebra $F_t$.

\end{definition}

This means that the value of $X_t$ depends only on the information available up to time $t$. In other words, the process ``adapts'' to the information in a way that it cannot anticipate future values based on events that have not occurred yet. We introduce the following locale.

\begin{isalemma}
{\small
\begin{lstlisting}[style=isabelle]
locale adapted_process = filtered_measure $M$ $F$ $t_0$ for $M$ $F$ $t_0$ 
	and $X$ :: "$\_ \; \Rightarrow \; \_ \; \Rightarrow \; \_ :: \; \{$second_countable_topology, banach$\}$" $+$
  assumes adapted[measurable]: "$\bigwedge i. \; t_0 \le i \implies X \; i \in \texttt{borel\_measurable} \; (F \; i)$"
\end{lstlisting}
}
\end{isalemma}

The properties we have shown concerning stochastic processes also hold for adapted processes. Although in some cases, we need to modify the measurability assumptions we make. For example, a constant process is adapted to a filtration, only if the defining function $f$ is $F_{t_0}$-measurable. The following formal lemma reflects this.
\begin{isalemma}
{\small
\begin{lstlisting}[style=isabelle]
lemma (in filtered_measure) adapted_process_const_fun:
  assumes "$f \in \texttt{borel\_measurable} \; (F \; t_0)$"
  shows "$\texttt{adapted\_process} \; M \; F \; t_0 \; (\lambda \_. \; f)$"
\end{lstlisting}
}
\end{isalemma}
\begin{remark}
Here, we see how constraining ourselves to an index set bounded from below helps make the assumption simpler. In order to show that the process is adapted, it suffices to check that $f$ is $F_{t_0}$-measurable. Without such a least element, we would have to check that $f$ is $F_t$-measurable for all indices $t$.
\end{remark}
An adapted process is necessarily a stochastic process. This follows directly from the fact that $F_t \subseteq \Sigma$ for all $t \ge t_0$. We capture this via the following sublocale relation.

\begin{isalemma}
{\small
\begin{lstlisting}[style=isabelle]
sublocale adapted_process $\subseteq$ stochastic_process 

sublocale nat_adapted_process $\subseteq$ nat_stochastic_process ..
sublocale real_adapted_process $\subseteq$ real_stochastic_process ..
\end{lstlisting}
}
\end{isalemma}

In the other direction, a stochastic process is always adapted to the natural filtration it generates:

\begin{isalemma}
{\small
\begin{lstlisting}[style=isabelle]
sublocale stochastic_process $\subseteq$ adapted_natural: 
	adapted_process $M$ "natural_filtration $M$ $t_0$ $X$" $t_0$ $X$ 
\end{lstlisting}
}
\end{isalemma}

Adapted processes are cruicial for defining martingales. A martingale is by definition an adapted process. In the following section, we will explore progressively measurable processes, even though they are not directly relevant to our formalization of martingales. This serves two purposes: first, to replicate the corresponding results on \textsf{mathlib}, and second, to establish a solid foundation for future projects to build upon.

\section{Progressively Measurable Processes}

The definition of a progressively measurable process is more intricate.

\begin{definition}
	Let $(F_t)_{t \in [t_0, \infty)}$ be a filtration of the measure space $M$. A stochastic process $(X_t)_{t \in [t_0, \infty)}$ is called progressively measurable (or simply \textit{progressive}) if, for every index $t \ge t_0$, the map $[t_0, t] \times \Omega \rightarrow E$ defined by $(i, w) \mapsto X_i(w)$ is measurable with respect to the $\sigma$-algebra $\mathcal{B}([t_0, t]) \otimes F_t$. Here $\mathcal{B}([t_0, t])$ denotes the Borel $\sigma$-algebra on $[t_0, t]$ induced by the order topology.
\end{definition}

The formalized version is as follows.

\begin{isadefinition}
{\small
\begin{lstlisting}[style=isabelle]
locale progressive_process = filtered_measure $M$ $F$ $t_0$ for $M$ $F$ $t_0$
	and $X$ :: "$\_ \; \Rightarrow \; \_ \; \Rightarrow \; \_ :: \; \{$second_countable_topology, banach$\}$" $+$
  assumes progressive[measurable]: "$\bigwedge t. \; t_0 \le t$
	$\implies (\lambda (i, x). \; X \; i \; x) \in \texttt{borel\_measurable} \; (\texttt{restrict\_space} \; \texttt{borel} \; \{t_0..t\} \otimes_M (F \; t))$"
\end{lstlisting}
}
\end{isadefinition}

Notice that the measurability assumption we make here is on the entire map $(i, w) \mapsto X_i(w)$ instead of being ``pointwise'' as in the previous two sections. As a side effect, the stochastic process defined by $X_i = c(i)$ for some $c : [t_0, \infty) \rightarrow E$ is progressively measurable, only if the function $c$ is Borel measurable. Previously, this assumption was not required. Hence, we modify the corresponding lemma:

\begin{isalemma}
{\small
\begin{lstlisting}[style=isabelle]
lemma (in filtered_measure) progressive_process_const:
  assumes "$c \in \texttt{borel\_measurable} \; \texttt{borel}$"
  shows "$\texttt{progressive\_process} \; M \; F \; t_0 \; (\lambda i \; \_. \; c \; i)$"
  \end{lstlisting}
}
\end{isalemma}

Similarly, we must modify the premise of the lemma \texttt{compose} in order to reflect this change:

\begin{isalemma}
{\small
\begin{lstlisting}[style=isabelle]
lemma compose_progressive:
  assumes "$(\lambda (i, x). \; f \; i \; x) \in \texttt{borel\_measurable} \; \texttt{borel}$"
  shows "$\texttt{progressive\_process} \; M \; F \; t_0 \; (\lambda i \; x. \; (f \; i) \; (X \; i \; x))$"
  \end{lstlisting}
}
\end{isalemma}

A progressively measurable process is necessarily adapted. The proof is trivial and arises from the fact that the injection $y \mapsto (t, y)$ is measurable as a function $\Omega \rightarrow [t_0, t] \times \Omega$ for fixed $t \ge t_0$. We formalize this fact as a sublocale relation.

\begin{isalemma}
{\small
\begin{lstlisting}[style=isabelle]
sublocale progressive_process $\subseteq$ adapted_process 
\end{lstlisting}
}
\end{isalemma}

On a more interesting note, progressive measurability is equivalent to adaptedness in the discrete-time setting. The following lemma demonstrates this.

\begin{lemma}
	Let $(X_i)_{i \in \mathbb{N}}$ be an adapted process with respect to the filtration $(F_i)_{i \in \mathbb{N}}$. Then it is also progressively measurable.
\end{lemma}
\begin{proof}
	Let $S$ be an open set in $E$. Then $X_j^{-1}(S) \in F_i$ for all $j \le i \in \mathbb{N}$, since $(X_i)_{i \in \mathbb{N}}$ is adapted by assumption. Let $\psi : \{0,\dots,i\} \times \Omega \rightarrow E$ with $\psi(j,x) = X_j(x)$. Then, we have
	\[
		\psi^{-1}(S) \cap \{j\} \times \Omega = \{j\} \times X_j^{-1}(S) \in \mathcal{B}(\{0,\dots,i\}) \otimes F_i
	\]
	since the order topology on $\mathbb{N}$ is discrete. Furthermore
	\[
		\psi^{-1}(S) = \bigcup_{j \le i} \psi^{-1}(S) \cap \{j\} \times \Omega 
	\]
	Since the set $\{0,\dots,i\}$ is countable, it follows that $\psi^{-1}(S) \in \mathcal{B}(\{0,\dots,i\}) \otimes F_i$, since it is expressable as the union of a countable family of measurable sets.
\end{proof}

Subsequently we express this fact in the language of locales.

\begin{isalemma}
{\small
\begin{lstlisting}[style=isabelle]
sublocale nat_adapted_process $\subseteq$ nat_progressive_process
\end{lstlisting}
}
\end{isalemma}

Now comes the most challenging portion of this chapter.

\section{Predictable Processes}

Before defining predictable processes in full generality, we will introduce them in the discrete-time setting, where the definition is easier to grasp.

\begin{definition}
	A discrete-time stochastic process $(X_i)_{i\in\mathbb{N}}$ is called \textit{predictable} with respect to a filtration $(F_i)_{i\in\mathbb{N}}$, if $X_{i + 1}$ is $F_i$-measurable for all $i \in \mathbb{N}$.
\end{definition}

This means that the value of the process in the future, $X_{i+1}$, can be ``predicted'' using the information available up to time $i$. This definition is a special case of the following more general definition for arbitrary index sets.

\begin{definition}
	Let $[t_0, \infty)$ be linearly ordered.  Let $(F_t)_{t\in[t_0, \infty)}$ be a filtration of the measure space $M$. We define the \textit{predictable $\sigma$-algebra} $\Sigma_P$ as follows.
	\[
		\Sigma_P = \sigma(\{(s,t] \times A \;\vert\; A \in F_s \;\wedge\; t_0 \le s \;\wedge\; s < t \} \cup \{\{t_0\} \times A \;\vert\; A \in F_{t_0}\})
	\]
	A stochastic process $(X_t)_{t\in[t_0, \infty)}$ is called \textit{predictable} if the map $[t_0, \infty) \times \Omega \rightarrow E$ defined by $(t,x) \mapsto X_t(x)$ is measurable with respect to this $\sigma$-algebra.
\end{definition}

At first glance, it is difficult to make intuitive sense of this definition. Investigating properties of predictable processes in arbitrary settings is well beyond the scope of this thesis. However, we will make the following remark.

\begin{remark}
	One can show that the $\sigma$-algebra $\Sigma_P$ coincides with the $\sigma$-algebra generated by all left-continuous adapted processes. A stochastic process is called left-continuous, if the sample paths $t \mapsto X_t(x)$ are left-continuous for $\mu$-almost all $x \in \Omega$. Right-continuity is similarly defined.
\end{remark}

The corresponding locale is easy to define.

\begin{isadefinition}
{\small
\begin{lstlisting}[style=isabelle]
locale predictable_process = linearly_filtered_measure $M$ $F$ $t_0$ for $M$ $F$ $t_0$
	and $X$ :: "$\_ \; \Rightarrow \; \_ \; \Rightarrow \; \_ :: \; \{$second_countable_topology, banach$\}$" $+$
  assumes progressive[measurable]: "$(\lambda (t, x). \; X \; t \; x) \in \texttt{borel\_measurable} \; \Sigma_P$"
\end{lstlisting}
}
\end{isadefinition}

In the previous section, our results concerning progressively measurable processes all made use of the fact that the projection functions $\pi_1 : [t_0, \infty) \times \Omega \rightarrow [t_0, \infty)$ and $\pi_2 : [t_0, \infty) \times \Omega \rightarrow \Omega$ are measurable with respect to the underlying $\sigma$-algebra. In that setting, this was a triviality, since the $\sigma$-algebra in question was the product $\sigma$-algebra $\mathcal{B}([t_0, t]) \otimes F_t$ for some $t \ge t_0$, which has many nice properties. We wish to show a similar statement for the projection functions $\pi_i$ when $[t_0, \infty) \times \Omega$ is equipped with the $\sigma$-algebra $\Sigma_P$. We have come up with a sufficient condition on the index set $[t_0, \infty)$ that guarantees this.

\begin{lemma}
	Assume there exists some countable family of sets $\mathcal{I} \subseteq \{(s, t] \;\vert\; t_0 \le s \;\wedge\; s < t\}$ such that $(t_0, \infty) \subseteq (\bigcup \mathcal{I})$. Let $\pi_1 : [t_0, \infty) \times \Omega \rightarrow [t_0, \infty)$ and $\pi_2 : [t_0, \infty) \times \Omega \rightarrow \Omega$ be projections onto respective components. Then, $\pi_1$ is $\Sigma_P$-Borel-measurable and $\pi_2$ is $\Sigma_P$-$F_{t_0}$-measurable.
\end{lemma}
\begin{proof}
	We first show that $\pi_1$ is $\Sigma_P$-Borel-measurable.
	
	$\pi_1$ is trivially $(\mathcal{B}([t_0,\infty)) \otimes \sigma(\varnothing))$-Borel-measurable. Hence, if we can show 
	\[
		(\mathcal{B}([t_0,\infty)) \otimes \sigma(\varnothing)) \subseteq \Sigma_P
	\]
	then this implies that $\pi_1$ is $\Sigma_P$-Borel-measurable. For this, we will show that the Borel $\sigma$-algebra $\mathcal{B}([t_0,\infty))$ coincides with the $\sigma$-algebra generated by the set $\{(s,t] \;\vert\; t_0 \le s \;\wedge\; s < t\}$.
	
	Since the ordering on $[t_0, \infty)$ is linear, the set of open rays $\{(s,\infty) \;\vert\; s \ge t_0\}$ generates the order topology on $\mathcal{B}([t_0,\infty))$. This also depends on the premise that the order topology is second-countable. Let $c \ge t_0$. We have
	\begin{align*}
		(c, \infty) &= (c, \infty) \cap (\bigcup \mathcal{I}) \\
		&= (\bigcup_{I \in \mathcal{I}} I \cap (c, \infty))
	\end{align*}
	From the assumptions, we know that 
	\[
		I \cap (c, \infty) \in \{(s, t] \;\vert\; t_0 \le s \;\wedge\; s < t\}
	\]
	for $I \in \mathcal{I}$. Hence $(c, \infty) \in \sigma(\{(s,t] \;\vert\; t_0 \le s \;\wedge\; s < t\})$, since $\mathcal{I}$ is countable.
	In the other direction, any interval $(s,t]$ with $s \ge t_0$ is obviously $\mathcal{B}([t_0,\infty))$-measurable. Thus, the $\sigma$-algebras indeed coincide. This completes the first part of the proof.
	
	Next, we show that $\pi_2$ is $\Sigma_P$-$F_{t_0}$-measurable.
	Let $S \in F_{t_0}$. We have
	\[
		\pi_2^{-1}(S) = [t_0, \infty) \times S
	\]
	The assumptions already imply $(t_0, \infty) = (\bigcup \mathcal{I})$. Furthermore $S \in F_t$ for all $t \ge t_0$. Hence, we have
	\[
		(t_0, \infty) \times S = (\bigcup \mathcal{I}) \times S \in \Sigma_P \;\;\textrm{and}\;\; \{t_0\} \times S \in \Sigma_P
	\]
	which together imply $[t_0, \infty) \times S \in \Sigma_P$.
\end{proof}

\begin{remark}
	Our formal proof for the $\Sigma_P$-Borel-measurability of $\pi_1$ follows an alternative path to the one given here. The lemma \texttt{borel\_Ioi} establishes that the Borel $\sigma$-algebra  $\mathcal{B}$ on the entire space is generated by open rays. We then consider the restricted $\sigma$-algebra $\mathcal{B}([t_0,\infty))$, which is defined in Isabelle as
	\[
		\mathcal{B}([t_0,\infty)) = \sigma\left(\left\{[t_0, \infty) \cap A \;\vert\; A \in \mathcal{B}\right\}\right)
	\]
	Together with \texttt{borel\_Ioi} this yields
	\[
	 	\mathcal{B}([t_0,\infty)) = \sigma\left(\left\{[t_0, \infty) \cap A \;\vert\; A \in \sigma(\{(s,\infty) \;\vert\; s \in (-\infty,\infty)\})\right\}\right)
	\]
	In our formalization, we show that this $\sigma$-algebra on the right hand side is equal to $\sigma(\{(s,t] \;\vert\; t_0 \le s \;\wedge\; s < t\})$.
	
	On a different note, we believe strongly that $\pi_2$ is not $\Sigma_P$-$F_t$-measurable for $t > t_0$ in general, or at least our condition is not sufficient to show this. This stems from the fact that $\pi_2^{-1}(S) = [t_0, \infty) \times S$ and the element $t_0$ can only originate from some set in $\{\{t_0\} \times A \;\vert\; A \in F_{t_0}\}$. However, in general $F_t \not\subseteq F_{t_0}$.
\end{remark}

In the discrete-time case, the family $\mathcal{I} = \{\{n + 1\}\}_{n \in \mathbb{N}}$ fulfills this condition. Similarly, in the continous-time case we can use $\mathcal{I} = \{(0,n + 1]\}\}_{n \in \mathbb{N}}$. In our formalization, we present these results in the context of the locales \texttt{nat\_filtered\_measure} and \texttt{real\_filtered\_measure}\footnote{\texttt{Filtered\_Measure.nat\_filtered\_measure.measurable\_predictable\_sigma\_snd} \quad \texttt{-.measurable\_predictable\_sigma\_fst} \\\texttt{real\_filtered\_measure.measurable\_predictable\_sigma\_snd} \quad \texttt{-.measurable\_predictable\_sigma\_fst} in \cite{Keskin_A_Formalization_of_2023}}.

These measurability results concerning projections are necessary to show the following statements about ``constant'' processes being predictable.

\begin{isalemma}
{\small
\begin{lstlisting}[style=isabelle]
lemma (in filtered_measure) predictable_process_const_fun:
  assumes "$\texttt{snd} \; \in \Sigma_P \rightarrow_M F \; t_0$" "$f \in \texttt{borel\_measurable} \; (F \; t_0)$"
  shows "$\texttt{predictable\_process} \; M \; F \; t_0 \; (\lambda\_. \; f)$"

lemma (in filtered_measure) predictable_process_const:
  assumes "$\texttt{fst} \; \in \texttt{borel\_measurable} \; \Sigma_P$" "$c \in \texttt{borel\_measurable} \; \texttt{borel}$"
  shows "$\texttt{predictable\_process} \; M \; F \; t_0 \; (\lambda i \_. \; c \; i)$"
\end{lstlisting}
}
\end{isalemma}

We will now show that a predictable process is necessarily progressively measurable.

\begin{lemma}
	A predictable process $(X_t)_{t \in [t_0,\infty)}$ is also progressively measurable.
\end{lemma}
\begin{proof}
	Let $i \ge t_0$. Let $\iota$ denote the identity function, restricted to the domain $[t_0, i] \times \Omega$, i.e. $\iota = \texttt{id}\vert_{[t_0,t]}$. We aim to show that $\iota$ is $(\mathcal{B}([t_0,i]) \otimes F_i)$-$\Sigma_P$-measurable. The statement follows simply from definitions of predictability and progressive measurability.
	
	For any $S$ in the generating set of $\Sigma_P$, we will show that $\iota^{-1}(S) \in \mathcal{B}([t_0,i]) \otimes F_i$. This is enough to show the required measurability.
	
	First, let $S = \{t_0\} \times A$ for some $A \in F_{t_0}$.
	Then
	\[
		\iota^{-1}(S) = \{t_0\} \times A \in \mathcal{B}([t_0,i]) \otimes F_i
	\]
	since $\{t_0\}$ is closed and $F_{t_0} \subseteq F_i$.
	
	Next, let $S = (s,t] \times A$ for some $A \in F_s$ and $s, t \ge t_0$ with $s < t$.
	Then
	\[
		\iota^{-1}(S) = (s, \min(i, t)] \times A
	\]
	
	Assume $s \le i$. Then $A \in F_i$. Furthermore, $(s, \min(i, t)] \in \mathcal{B}([t_0,i])$ since the $\sigma$-algebra $\mathcal{B}([t_0,i])$ is generated by half-open intervals. 
	Hence, $\iota^{-1}(S) \in \mathcal{B}([t_0,i]) \otimes F_i$.
	
	Assume $s > i$. Then $\iota^{-1}(S) = \varnothing \in \mathcal{B}([t_0,i]) \otimes F_i$. This covers all cases and the proof is complete.
\end{proof}

We formalize this fact as a sublocale relation.

\begin{isalemma}
{\small
\begin{lstlisting}[style=isabelle]
sublocale predictable_process $\subseteq$ progressive_process
\end{lstlisting}
}
\end{isalemma}

In the scope of our thesis, we will only use results concerning discrete-time predictable processes. We will now show that the general definition of a predictable process coincides with the definition in the discrete-time case. First we show the following lemma.

\begin{theorem}
	Let $(\bigcup_{i \in \mathbb{N}} \; \{i\} \times A_i) \in \Sigma_P$ for some collection of sets $(A_i)_{i \in \mathbb{N}}$. Then $A_0 \in F_0$ and $A_{i+1} \in F_i$ for all $i \in \mathbb{N}$.
\end{theorem}
\begin{proof}
	Consider the set 
	\[
		\mathcal{D} = \{S \in \Sigma_P \;\vert\; \forall (A_i)_{i \in \mathbb{N}}. \; S = (\bigcup_{i \in \mathbb{N}} \{i\} \times A_i) \Longrightarrow A_{i+1} \in F_i \;\wedge\; A_0 \in F_0 \}
	\]
	We will show that $\mathcal{D}$ constitutes a $\sigma$-algebra. Obviously $\varnothing \in \mathcal{D}$.
	
	Assume $S \in \mathcal{D}$. 
	
	Let $(A_i)_{i \in \mathbb{N}}$ be a family of sets with $(\mathbb{N} \times \Omega) \setminus S = (\bigcup_{i \in \mathbb{N}} \{i\} \times A_i)$. Then
	\begin{align*}
		S &= (\mathbb{N} \times \Omega) \setminus (\bigcup_{i \in \mathbb{N}} \{i\} \times A_i) \\
		&= (\bigcup_{i \in \mathbb{N}} \{i\} \times \Omega) \setminus (\bigcup_{i \in \mathbb{N}} \{i\} \times A_i) \\
		&= (\bigcup_{i \in \mathbb{N}} \{i\} \times (\Omega \setminus A_i))
	\end{align*}
	Hence, we know $\Omega \setminus A_{i+1} \in F_i$ and $\Omega \setminus A_0 \in F_0$. Therefore, $A_{i+1} \in F_i$ and $A_0 \in F_0$. We have $(\mathbb{N} \times \Omega) \setminus S \in \mathcal{D}$.
	
	Assume $S_i \in \mathcal{D}$ for $i \in \mathbb{N}$. 
	
	Let $(A_i)_{i \in \mathbb{N}}$ be a family of sets with $(\bigcup_{i \in \mathbb{N}} S_i) = (\bigcup_{i \in \mathbb{N}} \{i\} \times A_i)$. For each $S_i$, we need to find some family of sets $(B_j(i))_{j \in \mathbb{N}}$, such that $S_i = (\bigcup_{j \in \mathbb{N}} \{j\} \times B_j(i))$. Define
	\[
		B_j(i) = \pi_2(S_i \cap \{j\} \times \Omega)
	\]
	The intuition is as follows. We first select only those pairs in $S_i$, with the first component equal to $j$. Then we project onto the second component. Hence, we have the equality
	\[
		\{j\} \times B_j(i) = S_i \cap \{j\} \times \Omega
	\]
	Therefore $S_i = (\bigcup_{j \in \mathbb{N}} \{j\} \times B_j(i))$. We have $B_{j+1}(i) \in F_j$ and $B_0(i) \in F_0$. Furthermore, we know
	\begin{align*}
		A_i &= \pi_2\left((\bigcup_{j \in \mathbb{N}} S_j) \cap \{i\} \times \Omega\right) \\
		&= \bigcup_{j \in \mathbb{N}} \pi_2(S_j \cap \{i\} \times \Omega) \\
		&= \bigcup_{j \in \mathbb{N}} B_i(j)
	\end{align*}
	Hence, $A_{i + 1} \in F_i$ and $A_0 \in F_0$. Thus $\mathcal{D}$ is indeed a $\sigma$-algebra. Now we show 
	\[
		\{\{s + 1,\dots, t\} \times A \;\vert\; A \in F_s \;\wedge\; s < t \} \cup \{\{0\} \times A \;\vert\; A \in F_0\} \subseteq \mathcal{D}
	\]
	Let $S \in \{\{0\} \times A \;\vert\; A \in F_0\}$. Then $S = (\bigcup_{i \in \mathbb{N}} \{i\} \times A_i)$ implies $A_0 \in F_0$ and $A_i = \varnothing$ for $i > 0$. Hence $S \in \mathcal{D}$.
	
	Let $S = \{s + 1,\dots,t\} \times B$ with $s < t$ and $B \in F_s$ for some $s$, $t$ and $B$. Then, $S = (\bigcup_{i \in \mathbb{N}} \{i\} \times A_i)$ implies $A_i = B$ for $i \in \{s + 1,\dots,t\}$ and $A_i = \varnothing$ otherwise. Thus, $A_0 = \varnothing \in F_0$. Moreover, $A_{i + 1} = B \in F_i$ if $i \in \{s,\dots,t - 1\}$ since the subalgebras $F_i$ are nested, and $A_{i + 1} = \varnothing \in F_i$ for $i \notin \{s,\dots,t - 1\}$. Together with our previous result, this implies $\Sigma_P \subseteq \mathcal{D}$, which completes the proof.
\end{proof}

\begin{remark}
	For the proof of this lemma in Isabelle, we have used the induction scheme \texttt{sigma\_sets.induct}, since the generated $\sigma$-algebra $\sigma(\cdot)$ is defined as an inductive set in Isabelle. The proof above demonstrates that this induction scheme is equivalent to the principle of ``good-sets'' which we have utilized.
\end{remark}

We can now characterize predictability in the discrete-time setting as follows.
\begin{theorem}
	A stochastic process $(X_n)_{n \in \mathbb{N}}$ is predictable, if and only if $(X_{n + 1})_{n \in \mathbb{N}}$ is adapted to the filtration $(F_n)_{n \in \mathbb{N}}$ and $X_0$ is $F_0$-measurable.
\end{theorem}
\begin{proof}
	Assume $(X_n)_{n \in \mathbb{N}}$ is predictable. Since predictable processes are also adapted, $X_0$ is $F_0$-measurable.
	
	Let $n \in \mathbb{N}$ and let $S \subseteq E$ be an open set. Consider the map $\psi$ defined by $\psi(i,x) = X_i(x)$. We have
	\[
		\psi^{-1}(S) \cap (\{n + 1\} \times \Omega) = \psi^{-1}(S) \cap ((n, n + 1] \times \Omega) \in \Sigma_P
	\]
	On the other hand
	\[
		\psi^{-1}(S) \cap (\{n + 1\} \times \Omega) = \{n + 1\} \times X_{n + 1}^{-1}(S)
	\]
	Applying the previous lemma for 
	\[
		A_i = 
		\begin{cases}
			X_{n + 1}^{-1}(S) &\quad \textrm{if} \; i = n + 1 \\
			\varnothing &\quad \textrm{otherwise}
		\end{cases}
	\]
	we get $X_{n + 1}^{-1}(S) \in F_n$. Hence $(X_{n + 1})_{n \in \mathbb{N}}$ is adapted to the filtration $(F_n)_{n \in \mathbb{N}}$.
	
	For the other direction, assume $(X_{n + 1})_{n \in \mathbb{N}}$ is adapted to the filtration $(F_n)_{n \in \mathbb{N}}$ and $X_0$ is $F_0$-measurable.
	
	Let $S$ be an open set. We have
	\[
		\{0\} \times X_0^{-1}(S) \in \Sigma_P
	\]
	using the definition of $\Sigma_P$ and the fact that $X_0$ is $F_0$-measurable. Similarly, for $n \in \mathbb{N}$ we have $X_{n + 1}^{-1}(S) \in F_n$. Hence
	\[
		\{n + 1\} \times X_{n + 1}^{-1}(S) = (n, n + 1] \times X_{n + 1}^{-1}(S) \in \Sigma_P
	\]
	Putting it all together, we have
	\[
		\psi^{-1}(S) = \left(\bigcup_{i \in \mathbb{N}} \{i\} \times X_i^{-1}(S) \right) \in \Sigma_P
	\]
	since $\Sigma_P$ is closed under countable unions. Thus $(X_n)_{n \in \mathbb{N}}$ is predictable.
\end{proof}

This finalizes our formalization of various types of stochastic processes in Isabelle.

%% file: 05_martingales.tex

\chapter{Martingales}\label{chapter:martingales}

In this section we will introduce and discuss martingales, the namesake of our thesis. Originally referring to a system of betting strategies, martingales have evolved far beyond their gambling origins and have found profound applications in various fields, including finance, probability theory, and statistical analysis. Our formalization aims for a high level of generality while maintaining clarity and simplicity, making it easier for future formalization efforts to build upon our foundation. The results formalized in this chapter can be found in the theory file \texttt{Martingale.Martingale} in \cite{Keskin_A_Formalization_of_2023}.

\section{Fundamentals}

We start with the definition of a martingale.

\begin{definition}
	Let $(F_t)_{t \in [t_0,\infty)}$ be a filtration of the measure space $M$. A stochastic process $(X_t)_{t \in [t_0,\infty)}$ taking values in a Banach space $(E, \lVert \cdot \rVert)$ is a \textit{martingale with respect to the filtration $(F_t)_{t \in [t_0,\infty)}$} if the following conditions hold
	\begin{enumerate}
	\item $(X_t)_{t \in [t_0,\infty)}$ is adapted to the filtration $(F_t)_{[t_0,\infty)}$,
	\item $X_t \in L^1(E)$ for all $t \in [t_0, \infty)$,
	\item $X_s = \mathbb{E}(X_t \;\vert\; F_s)$ $\mu$-a.e. for all $s,t \in [t_0,\infty)$ with $s \le t$.
	\end{enumerate}
	Replacing ``$=$'' in the third condition with ``$\le$'' or ``$\ge$'' gives rise to the definition of a sub- or supermartingale, respectively.
\end{definition}

\begin{remark}
 In addition to what we have discussed in the last chapter, we have introduced the locale \texttt{sigma\_finite\_adapted\_process} which combines the locale \texttt{adapted\_process} with the locale \texttt{sigma\_finite\_filtered\_measure}. Without this additional restriction, we can not use the operator \texttt{cond\_exp}. Similary, the locale  \texttt{sigma\_finite\_adapted\_pro\-cess\-\_order} places a restriction on the Banach space $(E, \lVert \cdot \rVert)$, asserting the existence of an ordering compatible with scalar multiplication. Finally, the locale \texttt{sigma\_finite\_adap\-ted\-\_process\_linorder} further mandates that this ordering be total. We have also introduced locales for discrete-time and continuous-time counterparts.
\end{remark}

Using these additional definitions we introduce the following locale which formalizes martingales.

\begin{isadefinition}
{\small
\begin{lstlisting}[style=isabelle]
locale martingale = sigma_finite_adapted_process +
  assumes integrable: "$\bigwedge i. \; t_0 \le i \implies \texttt{integrable} \; M \; (X \; i)$"
      and martingale_property: "$\bigwedge i \; j. \; t_0 \le i \implies i \le j$
	  $\implies $AE$ \; x \; $in$ \; M. \; X \; i \; x = \texttt{cond\_exp} \; M \; (F \; i) \; (X \; j) \; x$"
\end{lstlisting}
}
\end{isadefinition}

Locales for submartingales and supermartingales are introduced similarly.

\begin{isadefinition}
{\small
\begin{lstlisting}[style=isabelle]
locale submartingale = sigma_finite_adapted_process_order +
  assumes integrable: "$\bigwedge i. \; t_0 \le i \implies \texttt{integrable} \; M \; (X \; i)$"
      and submartingale_property: "$\bigwedge i \; j. \; t_0 \le i \implies i \le j$
	  $\implies $AE$ \; x \; $in$ \; M. \; X \; i \; x \le \texttt{cond\_exp} \; M \; (F \; i) \; (X \; j) \; x$"
\end{lstlisting}
}
\end{isadefinition}

\begin{isadefinition}
{\small
\begin{lstlisting}[style=isabelle]
locale supermartingale = sigma_finite_adapted_process_order +
  assumes integrable: "$\bigwedge i. \; t_0 \le i \implies \texttt{integrable} \; M \; (X \; i)$"
      and supermartingale_property: "$\bigwedge i \; j. \; t_0 \le i \implies i \le j$
	  $\implies $AE$ \; x \; $in$ \; M. \; X \; i \; x \ge \texttt{cond\_exp} \; M \; (F \; i) \; (X \; j) \; x$"
\end{lstlisting}
}
\end{isadefinition}

Any stochastic process that is both a submartingale and a supermartingale is a martingale. Conversely, every martingale is also a submartingale and a supermartingale if there exists an ordering on the Banach space $E$. In anticipation of this result, we introduce the following locale.

\begin{isadefinition}
{\small
\begin{lstlisting}[style=isabelle]
locale martingale_order = martingale $M$ $F$ $t_0$ $X$ for $M$ $F$ $t_0$
	and $X$ :: "_ $\Rightarrow$ _ $\Rightarrow$ _ :: {order_topology, ordered_real_vector}"
\end{lstlisting}
}
\end{isadefinition}

Using thise locale we can state the following lemma which formalizes this fact.

\begin{isalemma}
{\small
\begin{lstlisting}[style=isabelle]
lemma martingale_iff: 
  shows "$\texttt{martingale} \; M \; F \; t_0 \; X \longleftrightarrow \texttt{submartingale} \; M \; F \; t_0 \; X \; \wedge \; \texttt{supermartingale} \; M \; F \; t_0 \; X$"
\end{lstlisting}
}
\end{isalemma}

Additionally, we have included lemmas for introducing martingales in simple cases. For $f \in L^1(E)$ and $F_{t_0}$-measurable, the constant stochastic process defined by $X_t = f$ is a martingale. The following lemma reflects this.

\begin{isalemma}
{\small
\begin{lstlisting}[style=isabelle]
lemma (in sigma_finite_filtered_measure) martingale_const_fun:  
  assumes "$\texttt{integrable} \; M \; f$" "$f \in \texttt{borel\_measurable} \; (F \; t_0)$"
  shows "$\texttt{martingale} \; M \; F \; t_0 \; (\lambda \_. \; f)$"
\end{lstlisting}
}
\end{isalemma}

The statements below follow directly.

\begin{isacorollary}
{\small
\begin{lstlisting}[style=isabelle]
corollary (in sigma_finite_filtered_measure) martingale_zero: 
	"$\texttt{martingale} \; M \; F \; t_0 \; (\lambda \_ \; \_. \; 0)$" by fastforce

corollary (in finite_filtered_measure) martingale_const: 
	"$\texttt{martingale} \; M \; F \; t_0 \; (\lambda \_ \; \_. \; c)$" by fastforce
\end{lstlisting}
}
\end{isacorollary}

Using our development of the conditional expectation operator, we have the following corollary.

\begin{corollary}
	The stochastic process defined by $X_t = \mathbb{E}(f \;\vert\; F_t)$ for $f \in L^1(E)$ is also a martingale. 
\end{corollary}
\begin{proof}
	This follows from the tower property of the conditional expectation.
\end{proof}

\section{Basic Operations and Alternative Characterizations}

First and foremost, we will discuss elementary properties of martingales, submartingales and supermartingales. 

\begin{lemma}
	Let $(X_t)_{t \in [t_0,\infty)}$ be a martingale with respect to a filtration $(F_t)_{t \in [t_0,\infty)}$. We have
	\begin{enumerate}
	\item[$\bullet$] Let $c \in \mathbb{R}$. Then $(c \cdot X_t)_{t \in [t_0,\infty)}$ is also a martingale.
	\item[$\bullet$] Let $(Y_t)_{t \in [t_0,\infty)}$ be another martingale with respect to the same filtration $(F_t)_{t \in [t_0,\infty)}$. Then $(X_t + Y_t)_{t \in [t_0,\infty)}$ and $(X_t - Y_t)_{t \in [t_0,\infty)}$ are also martingales.
	\end{enumerate}
\end{lemma}
\begin{proof}
	The statements follow using the linearity of the conditional expectation and the Bochner-integral.
\end{proof}

Similarly, we have the following statement for submartingales.

\begin{lemma}
	Let $(X_t)_{t \in [t_0,\infty)}$ be a submartingale with respect to a filtration $(F_t)_{t \in [t_0,\infty)}$. We have
	\begin{enumerate}
	\item[$\bullet$] Let $c \in \mathbb{R}$ with $c \ge 0$. Then $(c \cdot X_t)_{t \in [t_0,\infty)}$ is also a submartingale.
	\item[$\bullet$] Let $c \in \mathbb{R}$ with $c \le 0$. Then $(c \cdot X_t)_{t \in [t_0,\infty)}$ is a supermartingale.
	\item[$\bullet$] Let $(Y_t)_{t \in [t_0,\infty)}$ be another submartingale with respect to the same filtration $(F_t)_{t \in [t_0,\infty)}$. Then $(X_t + Y_t)_{t \in [t_0,\infty)}$ is also a submartingale.
	\end{enumerate}
\end{lemma}
\begin{proof}
	These statements also follow using the linearity of the conditional expectation and the Bochner-integral.
\end{proof}

\begin{remark}
In the lemma above we can exchange ``submartingale'' with ``supermartingale'' and the results are still valid.
\end{remark}

Going forward, we use the martingale property and the characterization of the conditional expectation to show the following lemma.

\begin{lemma}
	Let $(X_t)_{t \in [t_0,\infty)}$ be a martingale with respect to a filtration $(F_t)_{t \in [t_0,\infty)}$. Let $i,j \in [t_0,\infty)$ with $i \le j$ and $A \in F_i$. Then
	\[
		\int_A X_i \; \textrm{d}\mu = \int_A X_j \; \textrm{d}\mu
	\]
\end{lemma}

This lemma already shows us the intuition behind the definition of a martingale. Let $A$ be a set which is measurable at time $i$, i.e. some property of the process which we can inspect at time $i$. The average value that the process has on this set at time $i$ is equal to the average value it will have on the same set at a future time $j$. Essentially, this is the reason why martingales are employed for modeling fair games that incorporate an element of chance. 

Similarly, for submartingales we have the following lemmas.

\begin{lemma}
	Let $(X_t)_{t \in [t_0,\infty)}$ be a submartingale with respect to a filtration $(F_t)_{t \in [t_0,\infty)}$. Let $i,j \in [t_0,\infty)$ with $i \le j$ and $A \in F_i$. Then
	\[
		\int_A X_i \; \textrm{d}\mu \le \int_A X_j \; \textrm{d}\mu
	\]
\end{lemma}

Replacing ``$\le$'' with ``$\ge$'' gives a corresponding introduction lemma for supermartingales.

In this case, the intuition is similar. The average value of a submartingale on a set which is measurable at time $i$ is less than or equal to the average value it will take on the same set at a future time $j$. The case for a supermartingale is analogous. Here is a simple example illustrating this concept.

\begin{example}
Consider a coin-tossing game, where the coin lands on heads with probability $p \in [0,1]$. Assume that the gambler wins a fixed amount $c > 0$ on a heads outcome and loses the same amount $c$ on a tails outcome. Let $(X_n)_{n \in \mathbb{N}}$ be a stochastic process, where $X_n$ denotes the gambler's fortune after the $n$-th coin toss. Then, we have the following three cases.
\begin{enumerate}
\item If $p = \frac{1}{2}$, it means the coin is fair and has an equal chance of landing heads or tails. In this case, the gambler, on average, neither wins nor loses money over time. The expected value of the gambler's fortune stays the same over time. Therefore, $(X_n)_{n \in \mathbb{N}}$ is a martingale.
\item If $p \ge \frac{1}{2}$, it means the coin is biased in favor of heads. In this case, the gambler is more likely to win money on each bet. Over time, the gambler's fortune tends to increase on average. Therefore, $(X_n)_{n \in \mathbb{N}}$ is a submartingale.
\item If $p \le \frac{1}{2}$, it means the coin is biased in favor of tails. In this scenario, the gambler is more likely to lose money on each bet. Over time, the gambler's fortune decreases on average. Therefore, $(X_n)_{n \in \mathbb{N}}$ is a supermartingale.
\end{enumerate}
\end{example}

The property discussed above is of such fundamental significance that it can be employed to characterize martingales, submartingales, and supermartingales. We present the formal statement first, followed by a subsequent discussion of the proof idea.

\begin{lemma}
	Let $(X_t)_{t \in [t_0,\infty)}$ be an adapted process with respect to the filtration $(F_t)_{t \in [t_0,\infty)}$ consisting of Bochner-integrable random variables.
  Assume 
  \[
  \normalfont \int_A X_i \;\textrm{d}\mu = \int_A X_j \;\textrm{d}\mu
  \] 
  for all $A \in F_i$ and for all $i,j \in [t_0,\infty)$ with $i \le j$. Then the process $(X_t)_{t \in [t_0,\infty)}$ is a martingale.
\end{lemma}
\begin{proof}
Using the defining property of the conditional expectation, the second assumption can be restated as
\[
	\int_A X_i \;\textrm{d}\mu = \int_A \texttt{cond\_exp} \; M \; F_i \; X_j \;\textrm{d}\mu
\]
Applying the lemma on the uniqueness of densities (\ref{cor:density_unique}) and the fact that both functions are $F_i$-measurable, we get
\[
	\quad X_i = \texttt{cond\_exp} \; M \; F_i \; X_j \;\; \mu\vert_{F_i}\textrm{-a.e.}
\]
The statement follows from the fact that $F_i \subseteq \Sigma$, i.e. all $\mu\vert_{F_i}$-null sets are $\mu$-null sets. 
\end{proof}

Analogously, we have the following introduction lemma for submartingales.

\begin{lemma}
	Let $(X_t)_{t \in [t_0,\infty)}$ be an adapted process with respect to the filtration $(F_t)_{t \in [t_0,\infty)}$ consisting of Bochner-integrable random variables.
  Assume 
  \[
  \normalfont \int_A X_i \;\textrm{d}\mu \le \int_A X_j \;\textrm{d}\mu
  \] 
  for all $A \in F_i$ and for all $i,j \in [t_0,\infty)$ with $i \le j$. Then the process $(X_t)_{t \in [t_0,\infty)}$ is a submartingale.
\end{lemma}

The idea of the proof is the same. To prevent redundancy in our formalization, we have demonstrated the lemma for supermartingales\footnote{\texttt{Martingale.supermartingale\_of\_set\_integral\_ge} in \cite{Keskin_A_Formalization_of_2023}} by equivalently showing that the stochastic process $(-X_t)_{t \in [t_0,\infty)}$ is a submartingale using the lemma above\footnote{\texttt{Martingale.submartingale\_of\_set\_integral\_le} in \cite{Keskin_A_Formalization_of_2023}}. This is easily achieved using the locale system.

Another way to characterize martingales is by examining the conditional expectation of the difference between the process's values at different points in time. We have the following lemmas to address this.

\begin{lemma}
  Let $(X_n)_{n \in \mathbb{N}}$ be an adapted process with respect to the filtration $(F_n)_{n \in \mathbb{N}}$ consisting of Bochner-integrable random variables. Assume 
  \[
  \normalfont \texttt{cond\_exp} \; M \; (F_i) \; (X_{i+1} - X_i) = 0 \quad\mu\textrm{-a.e.}
  \] 
  Then the process $(X_n)_{n \in \mathbb{N}}$ is a martingale.
\end{lemma}

This characterization also has an informal interpretation. The expected future value of a martingale does not change, when we restrict it to those events that we can measure right now. In a similar fashion, the value of a submartingale is expected to increase and the value of a supermartingale is expected to decrease as time progresses.

\section{Discrete-Time Martingales}

Discrete-time martingales are widely used to model processes that evolve over a sequence of discrete time steps while satisfying the martingale property, such as financial asset prices, random walks, and stochastic games. We define the following locales.

\begin{isadefinition}
{\small
\begin{lstlisting}[style=isabelle]
locale nat_martingale = martingale $M$ $F$ "$0 :: \; \texttt{nat}$" $X$ for $M$ $F$ $X$
locale nat_submartingale = submartingale $M$ $F$ "$0 :: \; \texttt{nat}$" $X$ for $M$ $F$ $X$
locale nat_supermartingale = supermartingale $M$ $F$ "$0 :: \; \texttt{nat}$" $X$ for $M$ $F$ $X$
\end{lstlisting}
}
\end{isadefinition}

Many of the statements we have made above, can be simplified when the indexing set is the natural numbers. Given a point in time $i \in \mathbb{N}$, it suffices to consider the successor $i + 1$, instead of all future times $j \ge i$. This can be stated as follows.

\begin{lemma}
	Let $(F_n)_{n \in \mathbb{N}}$ be a filtration of the measure space $M$. A stochastic process $(X_n)_{n \in \mathbb{N}}$ which is adapted to the filtration $(F_n)_{n \in \mathbb{N}}$ is a martingale, if and only if it is integrable at all time indices $n \in \mathbb{N}$ and $X_n = \mathbb{E}(X_{n + 1} \;\vert\; F_n)$ $\mu$-a.e. for all $n \in \mathbb{N}$.
\end{lemma}
\begin{proof}
	The ``only if'' part is evident. For the other direction, let $n, m \in \mathbb{N}$ with $n \le m$. We will show that $X_n = \mathbb{E}(X_m \;\vert\; F_n)$ holds $\mu$-a.e. via induction on the difference $k := m - n$. 
	\begin{enumerate}
	\item[$\bullet$] For the induction basis, assume $k = 0$. Hence $n = m$. Since the process is adapted by our assumption, we have $X_n = \mathbb{E}(X_n \;\vert\; F_n) = \mathbb{E}(X_m \;\vert\; F_n)$.
	
	\item[$\bullet$] For the induction step, let $k + 1 = m - n$, and assume $X_i = \mathbb{E}(X_j \;\vert\; F_i)$ for all $i, j \in \mathbb{N}$ such that $k = j - i$ and $i \le j$. Using the fact that $k = m - (n + 1) \ge 0$, we have
	\[
		X_{n + 1} = \mathbb{E}(X_m \;\vert\; F_{n + 1}) \;\;\mu\textrm{-a.e.}
	\]
	We take the conditonal expectation of both sides with respect to the sub-$\sigma$-algebra $F_n$.
	\[
		\mathbb{E}(X_{n + 1} \;\vert\; F_n) =  \mathbb{E}(\mathbb{E}(X_m \;\vert\; F_{n + 1}) \;\vert\; F_n) \;\;\mu\textrm{-a.e.}
	\]
	Using the tower propery, we have
	\[
		\mathbb{E}(X_{n + 1} \;\vert\; F_n) =  \mathbb{E}(X_m \;\vert\; F_n) \;\;\mu\textrm{-a.e.}
	\]
	Finally, we use our assumption to get
	\[
		X_n =  \mathbb{E}(X_m \;\vert\; F_n) \;\;\mu\textrm{-a.e.}
	\]	
	which completes the proof by induction.
	\end{enumerate}
\end{proof}

The same proof idea can be used to show the statement for submartingales (resp. supermartingales) by replacing ``$=$'' with ``$\le$'' (resp. ``$\ge$'') and using the monotonicity of the conditional expectation. We use this characterization to introduce the following introduction lemmas for discrete-time martingales. Analogous statements hold for submartingales and supermartingales as well:

\begin{lemma}
  Let $(X_n)_{n \in \mathbb{N}}$ be an adapted process with respect to the filtration $(F_n)_{n \in \mathbb{N}}$ consisting of Bochner-integrable random variables. Assume 
  \[
  \normalfont X_i = \texttt{cond\_exp} \; M \; (F_i) \; X_{i+1} \quad\mu\textrm{-a.e.} 
  \]
  for all $i \in \mathbb{N}$. Then the process $(X_n)_{n \in \mathbb{N}}$ is a martingale.
\end{lemma}

We can express the alternative characterizations in the previous subsections this way as well:

\begin{lemma}
  Let $(X_n)_{n \in \mathbb{N}}$ be an adapted process with respect to the filtration $(F_n)_{n \in \mathbb{N}}$ consisting of Bochner-integrable random variables. 
  Assume 
  \[
  \normalfont \int_A X_i \;\textrm{d}\mu = \int_A X_{i+1} \;\textrm{d}\mu
  \] 
  for all $A \in F_i$ for all $i \in \mathbb{N}$. Then the process $(X_n)_{n \in \mathbb{N}}$ is a martingale.
\end{lemma}

\begin{lemma}
  Let $(X_n)_{n \in \mathbb{N}}$ be an adapted process with respect to the filtration $(F_n)_{n \in \mathbb{N}}$ consisting of Bochner-integrable random variables. Assume 
  \[
  \normalfont \texttt{cond\_exp} \; M \; F_i \; (X_{i+1} - X_i) = 0 \quad\mu\textrm{-a.e.}
  \] 
  for all $i \in \mathbb{N}$. Then the process $(X_n)_{n \in \mathbb{N}}$ is a martingale.
\end{lemma}

Lastly, we will discuss the behavior of discrete-time martingales, under the additional assumption that they are predictable. At the end of the preceeding chapter, we have shown that a discrete-time $(X_n)_{n \in \mathbb{N}}$ process is predictable, if and only if the time shifted process $(X_{n + 1})_{n \in \mathbb{N}}$ is an adapted process. That is, $X_{n+1}$ is $F_n$ measurable for all $n \in \mathbb{N}$. Under this additional assumption the martingale property becomes trivial, i.e.
\[
	X_n = \mathbb{E}(X_{n+1} \;\vert\; F_n) = X_{n+1} \;\;\mu\textrm{-a.e.}
\]
By induction, we have $X_n = X_0$ $\mu$-a.e. Hence, a predictable martingale must be constant. The formalized statement is as follows.

\begin{isalemma}
{\small
\begin{lstlisting}[style=isabelle]
lemma (in nat_martingale) predictable_const:
  assumes "$\texttt{nat\_predictable\_process} \; M \; F \; X$"
  shows "AE$ \; x \; $in$ \; M. \; X \; i \; x = X \; j \; x$"
\end{lstlisting}
}
\end{isalemma}

In the same vein, a predictable submartingale must be monotonically increasing and a predictable supermartingale must be monotonically decreasing:

\begin{isalemma}
{\small
\begin{lstlisting}[style=isabelle]
lemma (in nat_submartingale) predictable_mono:
  assumes "$\texttt{nat\_predictable\_process} \; M \; F \; X$" "$i \le j$"
  shows "AE$ \; x \; $in$ \; M. \; X \; i \; x \le X \; j \; x$"
\end{lstlisting}
}
\end{isalemma}

\begin{isalemma}
{\small
\begin{lstlisting}[style=isabelle]
lemma (in nat_supermartingale) predictable_mono:
  assumes "$\texttt{nat\_predictable\_process} \; M \; F \; X$" "$i \le j$"
  shows "AE$ \; x \; $in$ \; M. \; X \; i \; x \ge X \; j \; x$"
\end{lstlisting}
}
\end{isalemma}

This wraps up our formalization of martingales in Isabelle. As it was out of the scope of this thesis, we have not formalized any major results concerning martingales. That being said, we hold a strong belief that the groundwork we've laid here will provide an excellent foundation for future formalization endeavors.

%% file: 06_discussion.tex

\chapter{Discussion}\label{chapter:discussion}

In this chapter, we discuss the decisions we took throughout the formalization process. Furthermore, we compare our work with the existing formalization of martingales in Lean, and outline directions for future research.

While constructing the conditional expectation operator, we have come up with the following approach, which draws inspiration from the construction in \cite{Hytoenen_2016}. Both our approach, and the one in \cite{Hytoenen_2016} are based on showing that the conditional expectation is a contraction on some dense subspace of the space of functions $L^1(E)$.

In our approach, we start by constructing the conditional expectation explicitly for simple functions. Then we showed that the conditional expectation is a contraction on simple functions, i.e. $\lVert \mathbb{E}(s \vert F)(x) \rVert \le \mathbb{E}(\lVert s(x) \rVert \vert F)$ for $\mu$-almost all $x \in \Omega$ with $s : \Omega \rightarrow E$ simple and integrable. Using this, we were able to show that the conditional expectation of a convergent sequence of simple functions is again convergent. Finally, we showed that this limit exhibits the properties of a conditional expectation. This approach has the benefit of being straightforward and easy to implement, since we could make use of the existing formalization for real-valued functions. 

Now, we present the following alternative method, which is described in detail in \cite{Hytoenen_2016}.

One first shows that the conditional expectation exists for functions in $f \in L^2(\mathbb{R})$. Then one uses the fact that functions in $L^1(\mathbb{R})$ can be approximated by functions in $L^1(\mathbb{R}) \cap L^\infty(\mathbb{R})$ to obtain a conditional expectation operator $L^1(\mathbb{R})\rightarrow L^1(\mathbb{R})$. Then, one shows that this operator is bounded and positive. Hence, it can be extended to a bounded operator $L^1(E)\rightarrow L^1(E)$, which retains the properties of a conditional expectation. In more detail, one argues as follows. 

Let $F \subset \Sigma$, be a sub-$\sigma$-algebra. First, one shows that the subspace 
\[
	L^2(\mathbb{R}; F) := \{f \in L^2(\mathbb{R}) \;\vert\; f \; \textrm{is} \; F\textrm{-measurable}\}
\] is a closed and convex subset of $L^2(\mathbb{R})$. Using the Hilbert projection theorem, we obtain a projection $P : L^2(\mathbb{R}) \rightarrow L^2(\mathbb{R}; F)$. We then verify that the projected function $Pf$ satisfies the properties of a conditional expectation using the fact that projections are self-adjoint. 

Next, we show that the conditional expectation is contractive with respect to the $L^1$-norm. We know that $L^1(\mathbb{R}) \cap L^\infty(\mathbb{R}) \subseteq L^2(\mathbb{R})$ and $L^1(\mathbb{R}) \cap L^\infty(\mathbb{R})$ is dense in $L^1(\mathbb{R})$. Therefore, we can extend the conditional expectation operator, which is currently only defined for functions in $L^2(\mathbb{R})$ to a contraction $L^1(\mathbb{R}) \rightarrow L^1(\mathbb{R}; F)$. Then, it is straightforward to verify that this operator is indeed the conditional expectation. Finally, one shows that this operator is positive and extends it to a bounded operator $L^1(E)\rightarrow L^1(E; F)$ which still has the properties of a conditional expectation.

This is an elegant way of showing that the conditional expectation exists as a bounded operator on $L^1(E)$. Furthermore, one can easily extend this definition to functions in $L^p(E)$. Our approach also uses a similar argument; we show that the conditional expectation is a contraction on the dense subset of $L^1(E)$ generated by integrable simple functions.

Another reason why we did not employ this approach is because the only formalization of the Hilbert projection theorem in Isabelle/HOL is for complex vector spaces. Therefore we decided to take a simpler approach and construct the conditional expectation using mostly measure theoretical arguments. This also makes the proofs more accesible in our opinion.

To actually formalize this operator in Isabelle, we have decided to first create a predicate \texttt{has\_cond\_exp} which characterizes the existence of a conditional expectation. Then, we define the actual operator in terms of Hilbert's choice function. This has a couple of advantages. First and foremost, it allows us to state many of the properties of the conditional expectation operator without first needing to show that it actually exists. From a mathematical perspective, this doesn't change anything in a major way. However in practice, it makes the formalization much easier.

One of the short-comings of our formalization is how lemmas concerning ordered Banach spaces are developed. In many stages, we require that the ordering on the Banach spaces be linear. Otherwise, it doesn't necessarily follow that the sets $[a,\infty) = \{x \in E\;\vert\; a \le x\}$ and $(\infty, a] = \{x \in E\;\vert\; x \le a\}$ are closed. This is cruicial in many stages of our formalization. 

With this in mind, there are weaker restrictions we can place on the ordering that allow us to obtain the same results. A Banach space $(E, \lVert \cdot \rVert)$ equipped with a lattice ordering is called a ``Banach lattice'', if for any $a, b\in E$ the following implication holds.
\[
	\lvert a \rvert \le \lvert b \rvert \implies \lVert a \rVert \le \lVert b \rVert
\]
where $\lvert x \rvert := x \vee -x$. Under this weaker assumption, the aforementioned sets are still closed and one can still show the monotonicity results concerning the conditional expectation. This is actually how the formalization is done on \textsf{mathlib}. Even though it would be great to introduce Banach lattices and show these results in this more general setting, the limited scope of our thesis did not permit this. This is a subject we might will explore in future projects.

In our formalization of stochastic processes, we have decided to restrict the index set to $\{t \; \vert\; t \ge t_0\}$ for some $t_0$. The main advantage of this approach is that it lets us formalize continuous-time stochastic processes with ease. Concretely, we can set $t_0 = 0 \in \mathbb{R}$ to introduce a stochastic process $(X_t)_{t \in [0, \infty)}$. Without such an initial index, we would either need to define a new type for non-negative real numbers or have to set $X_t$ to some predefined value for all $t < 0$. On the other hand, if the existence of such an initial index proves to be a problem, we can always use the \texttt{option} type, ordered via $\texttt{None} \le \texttt{Some} \; t$ for all $t$ and $\texttt{Some} \; t_1 \le \texttt{Some} \; t_2 \iff t_1 \le t_2$ for all $t_1, t_2$. Then we can just set $F_{\texttt{None}} = \sigma(\varnothing)$ and forget about this additional constraint altogether. Apart from this, we have aimed to place as little restriction on the index set as possible. Initially, we had formalized stochastic processes with the additional assumption that the index set be linearly ordered. After careful inspection, we have noticed that this is only necessary for our definition of predictable processes.

\begin{remark} 
	If we assume that a filtration is indexed by a lattice $\mathcal{A}$ of subsets of an indexing set $I$, we can still meaningfully define predictable processes. However, in this setting the generating set for the preditable $\sigma$-algebra is much less intuitive at first glance. More information on the definition of predictability in this setting can be found in \cite{Ivanoff1993}.
\end{remark}

Our main motivation for this project was to port the \textsf{mathlib} formalization on martingales to Isabelle/HOL. We have fully accomplished this goal. In Appendix \ref{chapter:appendix}, we present the tables \ref{tab:martingale_theories}, \ref{tab:submartingale_theories} and \ref{tab:supermartingale_theories}, which contain the results of the \textsf{mathlib} document \texttt{probability.mar\-tingale.basic} \cite{Degenne_Ying_2022} with their counterparts in our formalization \cite{Keskin_A_Formalization_of_2023}. We have aimed to state all our results with at least the same level of generality as their \textsf{mathlib} counterparts. Furthermore, we were able to restate all results which contained the real numbers with an arbitrary linearly ordered $\mathbb{R}$-Banach space instead. In the future, we aim to further weaken the assumption to only require Banach lattices, which will make our formalization strictly more general than the one on \textsf{mathlib}.

As briefly stated at the end of the last chapter, we have primarily focused on laying the groundwork for formalizing martingales in arbitrary Banach spaces, with specific emphasis on extending the conditional expectation operator and generalizing specific concepts from Bochner integration. Building upon our formalization framework, the natural next step is to delve into the formalization of key martingale theorems and properties, such as the martingale convergence theorem, the optional stopping theorem, and the Doob decomposition, among others. 

Another direction for exploration is the further development of the theory of stochastic processes introduced in this work. This will lead us into more advanced territories, such as stochastic differential equations (SDEs), It\^o calculus, and the theory of semimartingales. It\^o calculus is used to rigorously define and solve SDEs. SDEs describe how quantities evolve over time when influenced by both deterministic trends and random fluctuations. This is particularly valuable in finance for modeling asset prices, interest rates, and other financial variables that exhibit inherent uncertainty. The ability to work with SDEs allows researchers to develop sophisticated pricing models for financial derivatives, assess risk accurately, and optimize investment strategies.

%% file: 07_conclusion.tex

\chapter{Conclusion}\label{chapter:conclusion}

This thesis has been dedicated to the formalization of martingales in arbitrary Banach spaces, using the proof assistant Isabelle. Our central objective was to provide a rigorous foundation for this endeavor and expand upon existing formalizations. As we conclude our work, we reflect upon the contributions we have made.

A major achievement of our work is the extension of the conditional expectation operator from the familiar real-valued setting to the broader context of Banach spaces. This generalization allows for a more versatile and comprehensive approach to modeling stochastic processes, accommodating a wider range of applications. Furthermore, we have lifted many of the commonly used properties of the conditional expectation to this more general setting. We have introduced locale definitions to characterize various types of stochastic processes and filtered measure spaces, which was essential for the development of martingale theory in Banach spaces. We have introduced corresponding locales for discrete-time and continuous-time processes. Additionally, we have characterized predictable processes in the discrete-time setting, using proof methods from measure theory.

Finally, we have introduced suitable locales for martingales, demonstrated their basic properties, and presented alternative characterizations in the discrete-time setting. As we look ahead, we face several potential directions for future formalization endeavors, such as martingale convergence theorems, It\^o calculus and the theory of semimartingales. By providing a robust framework for martingales and related concepts in Banach spaces, we hope to facilitate further exploration and development in this field. Prior to our work, there was no development of martingales in a general Banach space setting within the \textsf{AFP}. With our contributions, this gap has been addressed, offering a solid starting point for further formalizations within the theory of stochastic processes.

%% file: appendix.tex
\newgeometry{left=0.5cm,right=0.5cm}
\chapter{Appendix}\label{chapter:appendix}

The following tables contain the results of the \textsf{mathlib} document \texttt{probability.mar\-tingale.basic} \cite{Degenne_Ying_2022} matched with the corresponding results from our formalization \cite{Keskin_A_Formalization_of_2023}.

\vspace{1cm}

\begin{longtable}{| p{.38\textwidth} p{.57\textwidth} |}
	\caption[Lookup Table for Martingale Lemmas and Definitions]{Lookup table for martingale lemmas and definitions}\label{tab:martingale_theories} \vspace{0.5cm} \\
	\hline
	\textsf{Lean} & \textsf{Isabelle} \\ \hline
	\texttt{martingale} & \texttt{martingale (locale)}  \\
	\texttt{martingale.adapted} & \texttt{adapted\_process.adapted}  \\
	\texttt{martingale.add} & \texttt{martingale.add}  \\
	\texttt{martingale.condexp\_ae\_eq} & \texttt{martingale.martingale\_property}  \\
	\texttt{martingale.eq\_zero\_of\_predictable} & \texttt{martingale.predictable\_const}  \\
	\texttt{martingale.integrable} & \texttt{martingale.integrable}  \\
	\texttt{martingale.neg} & \texttt{martingale.uminus}  \\
	\texttt{martingale.set\_integral\_eq} & \texttt{martingale.set\_integral\_eq}  \\
	\texttt{martingale.smul} & \texttt{martingale.scaleR}  \\
	\texttt{martingale.strongly\_measurable} & \texttt{stochastic\_process.random\_variable}  \\
	\texttt{martingale.sub} & \texttt{martingale.diff}  \\
	\texttt{martingale.submartingale} & \textsf{via sublocale relation}  \\
	\texttt{martingale.supermartingale} & \textsf{via sublocale relation}  \\
	\texttt{martingale\_condexp} & \texttt{sigma\_finite\_filtered\_measure.martingale\_cond\_exp}  \\
	\texttt{martingale\_const} & \texttt{finite\_filtered\_measure.martingale\_const}  \\
	\texttt{martingale\_const\_fun} & \texttt{sigma\_finite\_filtered\_measure.martingale\_const}  \\
	\texttt{martingale\_iff} & \texttt{martingale\_iff}  \\
	\texttt{martingale\_nat} & \texttt{nat\_sigma\_finite\_adapted\_process.martingale\_nat}  \\
	\texttt{martingale\_of\_condexp\_sub\_eq\_zero\_nat} & \texttt{nat\_sigma\_finite\_adapted\_process.martingale\_of\_cond\_exp\-\_diff\_Suc\_eq\_zero}  \\
	\texttt{martingale\_of\_set\_integral\_eq\_succ} & \texttt{nat\_sigma\_finite\_adapted\_process.martingale\_of\_set\_integ\-ral\_eq\_Suc}  \\
	\texttt{martingale\_zero} & \texttt{sigma\_finite\_filtered\_measure.martingale\_zero} \\
	\hline
\end{longtable}
\pagebreak
\begin{longtable}{| p{.38\textwidth} p{.57\textwidth} |}
	\caption[Lookup Table for Submartingale Lemmas and Definitions]{Lookup table for submartingale lemmas and definitions}\label{tab:submartingale_theories} \vspace{0.5cm} \\
	\hline
	\textsf{Lean} & \textsf{Isabelle} \\ \hline
	\texttt{submartingale} & \texttt{submartingale (locale)}  \\
	\texttt{submartingale.adapted} & \texttt{adapted\_process.adapted}  \\
	\texttt{submartingale.add} & \texttt{submartingale.add}  \\
	\texttt{submartingale.add\_martingale} & \texttt{submartingale.add}  \\
	\texttt{submartingale.ae\_le\_condexp} & \texttt{submartingale\_property}  \\
	\texttt{submartingale.condexp\_sub\_nonneg} & \texttt{submartingale.cond\_exp\_diff\_nonneg}  \\
	\texttt{submartingale.integrable} & \texttt{submartingale.integrable}  \\
	\texttt{submartingale.neg} & \texttt{submartingale.uminus}  \\
	\texttt{submartingale.pos} & \texttt{submartingale.max\_0}  \\
	\texttt{submartingale.set\_integral\_le} & \texttt{submartingale\_linorder.set\_integral\_le}  \\
	\texttt{submartingale.smul\_nonneg} & \texttt{submartingale.scaleR\_nonneg}  \\
	\texttt{submartingale.smul\_le\_zero} & \texttt{submartingale.scaleR\_le\_zero}  \\
	\texttt{submartingale.strongly\_measurable} & \texttt{stochastic\_process.random\_variable}  \\
	\texttt{submartingale.sub\_martingale} & \texttt{submartingale.diff}  \\
	\texttt{submartingale.sub\_supermartingale} & \texttt{submartingale.diff}  \\
	\texttt{submartingale.sum\_mul\_sub} & \texttt{nat\_submartingale.partial\_sum\_scaleR}  \\
	\texttt{submartingale.sum\_mul\_sub'} & \texttt{nat\_submartingale.partial\_sum\_scaleR'}  \\
	\texttt{submartingale.sup} & \texttt{submartingale\_linorder.max}  \\
	\texttt{submartingale.zero\_le\_of\_predictable} & \texttt{nat\_submartingale.predictable\_mono}  \\
	\texttt{submartingale\_nat} & \texttt{nat\_sigma\_finite\_adapted\_process\_linorder.submartingale\-\_nat}  \\
	\texttt{submartingale\_of\_condexp\_sub\_nonneg} & \texttt{sigma\_finite\_adapted\_process\_order.submartingale\_of\_cond\-\_exp\_diff\_nonneg}  \\
	\texttt{submartingale\_of\_condexp\_sub\_nonneg\_nat} & \texttt{nat\_sigma\_finite\_adapted\_process\_linorder.submartingale\-\_of\_cond\_exp\_diff\_Suc\_nonneg}  \\
	\texttt{submartingale\_of\_set\_integral\_le} & \texttt{sigma\_finite\_adapted\_process\_linorder.submartingale\_of\-\_set\_integral\_le}  \\
	\texttt{submartingale\_of\_set\_integral\_le\_succ} & \texttt{nat\_sigma\_finite\_adapted\_process\_linorder.submartingale\-\_of\_set\_integral\_le\_Suc} \\
	\hline
\end{longtable}
\pagebreak
\begin{longtable}{| p{.38\textwidth} p{.57\textwidth} |}
	\caption[Lookup Table for Supermartingale Lemmas and Definitions]{Lookup table for supermartingale lemmas and definitions}\label{tab:supermartingale_theories} \vspace{0.5cm} \\
	\hline
	\textsf{Lean} & \textsf{Isabelle} \\ \hline
	\texttt{supermartingale} & \texttt{supermartingale (locale)}  \\
	\texttt{supermartingale.adapted} & \texttt{adapted\_process.adapted}  \\
	\texttt{supermartingale.add} & \texttt{supermartingale.add}  \\
	\texttt{supermartingale.add\_martingale} & \texttt{supermartingale.add}  \\
	\texttt{supermartingale.condexp\_ae\_le} & \texttt{supermartingale\_property}  \\
	\texttt{supermartingale.integrable} & \texttt{supermartingale.integrable}  \\
	\texttt{supermartingale.le\_zero\_of\_predictable} & \texttt{supermartingale.predictable\_mono}  \\
	\texttt{supermartingale.neg} & \texttt{supermartingale.uminus}  \\
	\texttt{supermartingale.set\_integral\_le} & \texttt{supermartingale\_linorder.set\_integral\_ge}  \\
	\texttt{supermartingale.smul\_nonneg} & \texttt{supermartingale.scaleR\_nonneg}  \\
	\texttt{supermartingale.smul\_le\_zero} & \texttt{supermartingale.scaleR\_le\_zero}  \\
	\texttt{supermartingale.strongly\_measurable} & \texttt{stochastic\_process.random\_variable}  \\
	\texttt{supermartingale.sub\_martingale} & \texttt{supermartingale.diff}  \\
	\texttt{supermartingale.sub\_submartingale} & \texttt{supermartingale.diff}  \\
	\texttt{supermartingale\_nat} & \texttt{nat\_sigma\_finite\_adapted\_process\_linorder.supermartingale\-\_nat}  \\
	\texttt{supermartingale\_of\_condexp\_sub\_nonneg\-\_nat} & \texttt{nat\_sigma\_finite\_adapted\_process\_linorder.supermartingale\-\_of\_cond\-\_exp\_diff\_Suc\_le\_zero}  \\
	\texttt{supermartingale\_of\_set\_integral\_succ\_le} & \texttt{nat\_sigma\_finite\_adapted\_process\_linorder.supermartingale\-\_of\_set\_integral\_le\_Suc} \\
	\hline
\end{longtable}

\restoregeometry